\documentclass[a4paper, 10pt, conference]{ieeeconf}      

\IEEEoverridecommandlockouts                              

\usepackage{cite}
\usepackage{amsmath,amssymb,amsfonts}
\usepackage{algorithmic}
\usepackage{textcomp}
\usepackage{mathtools}
\usepackage{subcaption}
\usepackage{graphicx}
\usepackage[ruled, lined, longend]{algorithm2e}
\usepackage{enumerate}
\usepackage{multirow}
\usepackage{hyperref}
\newtheorem{definition}{Definition}
\newtheorem{assumption}{Assumption}
\newtheorem{theorem}{Theorem}
\newtheorem{remark}{Remark}
\newtheorem{proposition}{Proposition}

\newtheorem{lemma}{Lemma}
\newtheorem{eg}{Example}

\allowdisplaybreaks
\begin{document}
\title{Robust Output-Lifted Learning Model Predictive Control}

\author{Siddharth H. Nair,  Francesco Borrelli
\thanks{SHN, FB are with the Department of Mechanical Engineering, University of California Berkeley. Email IDs: \{siddharth\_nair, fborrelli\}@berkeley.edu}
}
\maketitle
\thispagestyle{empty}
\pagestyle{empty}

\begin{abstract}
We propose an iterative approach for designing Robust Learning Model Predictive Control (LMPC) policies for a class of nonlinear systems with additive, unmodelled dynamics. The nominal dynamics are assumed to be difference flat, i.e., the  state and input can be reconstructed using flat output sequences. For the considered class of systems, we synthesize Robust MPC policies and show how to use historical trajectory data collected during iterative tasks to 1) obtain bounds on the unmodelled dynamics and 2) construct a convex value function approximation along with a convex safe set in the space of output sequences for designing terminal components in the Robust MPC design. We show that the proposed strategy guarantees robust constraint satisfaction, asymptotic convergence to a desired subset of the state space and non-decreasing closed-loop performance at each policy update. Finally, simulation results demonstrate the effectiveness of the proposed strategy on a minimum time control problem  using a constrained nonlinear and uncertain vehicle model.

\end{abstract}
\section{Introduction}
Infinite-horizon optimal control has a long and celebrated history, with the cornerstones laid in the 1950s by  \cite{pontryagin2018mathematical} and \cite{bellman1966dynamic}. The problem involves seeking a control signal that minimizes the cost incurred by a trajectory of a dynamical system starting from an initial condition over an infinite time horizon. While certain problem settings admit analytical solutions (like unconstrained LQR \cite{kwakernaak1972linear}), the infinite-horizon optimal control problem for general nonlinear dynamical systems subject to constraints, is challenging to solve. This is because these problems require the numerical solution of an infinite-dimensional optimization problem, which is intractable even in the discrete-time setting (where the solution is an infinite sequence of control inputs instead of a control input signal).

Model Predictive Control (MPC) 
is an attractive methodology for tractable synthesis of feedback control of constrained nonlinear discrete-time systems. The control action at every instant requires the solution of a finite-horizon optimal control problem with a suitable constraint and cost on the terminal state of the system to approximate the infinite-horizon problem. These terminal components are designed so that the closed-loop system is stablized to a desired \textit{goal set} and satisfies constraints. This is achieved by constraining the terminal state to lie in a control invariant set containing the goal set, with an associated Control Lyapunov function (CLF). 
The computation of these sets with an accompanying CLF for nonlinear systems is challenging in general, and typically require local approximations of the nonlinear dynamics around the goal set. A  proper review of such constructions goes outside the scope of this article.


For iterative tasks where the system starts from the same position for every iteration of the optimal control problem, data from previous iterations may be used to update the MPC design using ideas from Iterative Learning Control (ILC)~\cite{ILCsurvey,ILC_MPC_cost, ILC_MPC_model}.
In these strategies the goal of the controller is to track a given reference trajectory, and the tracking error from the previous execution is used to update the controller. For control problems where a reference trajectory may be hard to compute, \cite{UgoTAC} proposed a
reference-free iterative policy synthesis strategy, called Learning Model Predictive Control (LMPC) which iteratively constructs a control invariant terminal set and an accompanying terminal cost function using historical data. These quantities are discrete, therefore the LMPC relies on the solution of a Mixed-Integer Nonlinear Program (MINLP) at each instant for guaranteed stability and constraint satisfaction. 
In \cite{nair2021output}, we build on the work of~\cite{UgoTAC} and proposed a strategy to reduce the computational burden of LMPC for a class of nonlinear systems by replacing these discrete sets and functions with continuous, convex ones while still maintaining safety and performance  guarantees. Moreover, these quantities are computed without any local approximations of the nonlinear dynamics.

The LMPC framework of \cite{nair2021output} and this paper, considers discrete-time nonlinear systems for which the state and input can be reconstructed using certain system output sequences, which are defined as \textit{lifted outputs}.  These outputs sequences are constructed using \textit{flat} outputs \cite{flat_def} which have also been used in \cite{aranda1996linearization} to construct dynamic feedback linearizing inputs for discrete-time systems. Existing work on constrained control for such systems require a carefully designed reference trajectory which is then tracked using MPC with a linear model obtained either by a first order approximation \cite{de2009flatness} or by feedback linearization \cite{wang2019flatness,greeff2018flatness, kandler2012differential, alsalti2022data, agrawal2021constructive}. In both cases, there are no formal guarantees of closed-loop system stability and constraint satisfaction, except in \cite{alsalti2022data} and \cite{agrawal2021constructive}. \cite{alsalti2022data} proposes a model-free, data-driven approach based on Willem's fundamental lemma for Robust MPC by using input-output data of the feedback-linearized system. However, the formulation can not enforce state constraints for a recursively feasible Robust MPC scheme. In \cite{agrawal2021constructive} a model-based hierarchical approach is considered: a Robust MPC scheme with the feedback-linearized dynamics provides a reference trajectory, and a low level tracking controller for original nonlinear system is designed to track the reference trajectory. The approach in \cite{agrawal2021constructive} does not address input constraints, and the terminal set is chosen as the desired \textit{goal set}.  In this article, we design a model-based and data-driven Robust LMPC framework for the class of difference flat nonlinear systems, with additive uncertainty to capture unmodelled dynamics. The main contributions of this article are as follows:
\begin{enumerate}\item We use historical trajectory data to quantify and iteratively decrease the model uncertainty via semi-definite programming using techniques from \cite{nair2020modeling, bujarbaruah2020semi}, generalized beyond Lipschitz nonlinearities. This is used for construction of tightened state and input constraints.

\item We iteratively construct convex terminal sets and terminal costs in the space of lifted outputs using historical trajectory data for the Robust MPC optimization problem.
\item Finally, we show that the proposed Robust LMPC strategy in closed-loop ensures $i$) constraint satisfaction, $ii$) convergence to a desired set, $iii$) non-decreasing closed-loop system performance across iterations.
\end{enumerate}

The paper is organized as follows. We begin by formally describing the problem we want to solve in Section~\ref{sec:PD} along with required definitions. Section~\ref{sec:ROLMPC} details the construction of various  components of our the control design. The closed-loop properties are analysed in Section~\ref{sec:LMPC_analysis}. Finally, Section~\ref{sec:ex} presents numerical results that illustrate our proposed approach for optimal control of a kinematic bicycle.
\section{Problem Formulation}\label{sec:PD}
\subsection{System Model and Uncertainty Description}
Consider a nonlinear discrete-time system given by the dynamics
\begin{align}\label{eq:sysdyn}
    x_{t+1}=f(x_t,u_t)+d_t,\ d_t\in D(x_t,u_t)
\end{align}
 where $x_t\in \mathbb{R}^n$ and $u_t\in\mathbb{R}^m$  are the system state and input respectively at time $t$, and $f(\cdot,\cdot)$ is a known, continuous function.  The disturbance $d_t$ is assumed to belong to compact set $D(x_t, u_t)$, where $D(\cdot)$ is a set-valued map, $D:\mathbb{R}^{n+m}\rightarrow 2^{\mathbb{R}^n}$. The map $D(\cdot)$ is \textit{unknown} and represents unmodelled dynamics. We assume that this map satisfies the incremental property stated by the following assumption.
 \begin{assumption}\label{ass:dQC}
 The unknown set-valued map $D(\cdot)$ satisfies the following quadratic constraint, for any $z=(x,u), z'=(x',u')$ in $\mathbb{R}^{n+m}$:
 \begin{align*}
     &\begin{bmatrix}
     1\\z-z'\\d-d'
     \end{bmatrix}^\top Q^{(j)}\begin{bmatrix}
     1\\z-z'\\ d-d'
     \end{bmatrix}\geq 0, \begin{aligned}\\\forall d\in D(z),\forall d'\in D(z'),\\
     \forall Q^{(j)}\in\boldsymbol{\mathcal{Q}}~~~~~~~~~\end{aligned}
 \end{align*}
 where $\boldsymbol{\mathcal{Q}}\subset\mathbb{R}^{2n+m+1\times 2n+m+1}$\ is a known, finite set of symmetric matrices.
 \end{assumption}
 Each matrix $Q^{(j)}$ in $\boldsymbol{\mathcal{Q}}$, captures side information on the unmodelled dynamics such as sector bounds, Lipschitz constants or Jacobian bounds \cite{megretski1997system, hashemi2021certifying}.
 \begin{remark}\label{rem:approx_constants}
 The side information
are typically approximated using data, or chosen as \textit{hyperparameters} to be tuned. In the former case, without additional assumptions, robust constraint satisfaction can be guaranteed only in probability~\cite{maddalena2021deterministic, manzano2020robust, nubert2020safe, koller2018learning, vinod2022fly}.
 \end{remark}

 \begin{eg}\label{eg:lip_bnded}
 \textit{Suppose that the disturbance takes the form $d_t=w_t+d(x_t,u_t)$, where $w_t$ lies within the set $\mathcal{W}:=\{w| \Vert w\Vert_2\leq\gamma\}$, $d(x,u)$ is an unknown, $L-$Lipschitz continuous function. and so $$D(x_t,u_t)=\{d_t\in\mathbb{R}^n\ |\ \exists w_t\in\mathcal{W}, d_t=w_t+d(x_t,u_t)\}.$$ Then for any $z=(x,u),z'=(x',u')$, $w_1,w_2\in\mathcal{W}$, we have (using the Lipschitz constant $L$ of $d(\cdot,\cdot)$ and bound $\gamma$ for $w_t$) that $||d(z)+w_1-d(z')-w_2\Vert^2_2\leq (L\Vert z-z'\Vert_2+2\gamma)^2\leq 2L^2\Vert z-z'\Vert^2_2+8\gamma^2$. Thus $\boldsymbol{\mathcal{Q}}$ consists of a single matrix $Q=\text{blkdiag}([8\gamma^2,2L^2 I_n, -I_n])$}.
 \end{eg}
 
 Also define the \textit{nominal} nonlinear discrete-time system,
 \begin{align}\label{eq:nom_dyn}
     \bar{x}_{t+1}&=f(\bar{x}_t,\bar{u}_t)\nonumber\\
     \bar{y}_t&=h(\bar{x}_t)
 \end{align}
 where $\bar{x}_t\in\mathbb{R}^n$, $\bar{u}_t, \bar{y}_t \in\mathbb{R}^m$,  are the nominal system state, input and output at time $t$. The output $\bar{y}_t$ is a difference flat output (\cite{flat_def}), and is used to construct the \textit{lifted} output for the nominal system \eqref{eq:nom_dyn} as discussed next.
\begin{definition}\label{def:diffFlat} Let $\bar{y}_t=h(\bar{x}_t)$ with $h : \mathbb{R}^n\rightarrow \mathbb{R}^m$ be the output of system \eqref{eq:nom_dyn}. If $\exists R\in\mathbb{N}$ and a function $\mathcal{F}:\mathbb{R}^{m\times R+1}\rightarrow \mathbb{R}^n\times\mathbb{R}^m$,
such that any state/input pair ($\bar{x}_t$, $\bar{u}_t$) can be uniquely reconstructed from a sequence of outputs $\bar{y}_t,\dots, \bar{y}_{t+R}$ as
\begin{align}\label{eq:flatdef}
    (\bar x_{t},\bar u_{t})&=\mathcal{F}([\bar y_{t},\bar y_{t+1},\dots,\bar y_{t+R}]),
\end{align}
then the lifted output is the matrix \begin{align}\label{eq:vrtl_lftd_op}
\bar{\mathbf{Y}}_t=[\bar{y}_t,\dots,\bar{y}_{t+R}]\in\mathbb{R}^{m\times R+1}.\end{align}
\end{definition}
We formally assume the existence of the lifted output for our nominal system along with some additional structure on the map $\mathcal{F}(\cdot)$ next.

\begin{assumption}\label{ass:flat}
We are given an output function $\bar y_t=h(\bar x_t)$ with corresponding lifted output $\bar{\mathbf{Y}}_t$ for the nominal system \eqref{eq:nom_dyn}. Moreover, the map $\mathcal{F}(\cdot)$ in~\eqref{eq:flatdef} also satisfies the following properties:
\begin{enumerate}[\label=(A)]
    \item \label{ass:flat_class}$\mathcal{F}(\cdot)$ is continuous, and requires $R$ and $R+1$ outputs for identifying the nominal state and nominal input respectively, i.e.,
\begin{align}
    \bar x_t&=\mathcal{F}_x([\bar y_t,\bar y_{t+1},\dots,\bar y_{t+R-1}])\label{eq:flatclass_x}\\
    \bar u_t&=\mathcal{F}_u([\bar y_t,\bar y_{t+1},\dots,\bar y_{t+R}])\label{eq:flatclass_u}
\end{align}
\item\label{ass:flatmap_convex}Let $\mathcal{F}^i: \mathbb{R}^{m\times R+1}\rightarrow \mathbb{R}$ be the $i$th component of the map $\mathcal{F}:\mathbb{R}^{m\times R+1}\rightarrow \mathbb{R}^n\times\mathbb{R}^m$ where $i=1,\dots, n+m$. For each $i\in\{1,\dots, n+m\}$, there exist functions $\mathcal{F}^{i,\cup}$, $\mathcal{F}^{i,\cap}$  such that $\mathcal{F}^{i,\cap}(\bar{\mathbf{Y}})\leq\mathcal{F}^i(\bar{\mathbf{Y}})\leq \mathcal{F}^{i,\cup}(\bar{\mathbf{Y}})$  where $\mathcal{F}^{i,\cap}$ is quasiconcave and $\mathcal{F}^{i,\cup}$ is quasiconvex, i.e.,
\small
\begin{align*}
&\forall\bar{\mathbf{Y}}^1,\bar{\mathbf{Y}}^2\in\mathbb{R}^{m\times R+1}, \forall t\in[0,1]:\\
&\min\{\mathcal{F}^{i,\cap}(\bar{\mathbf{Y}}^1),\mathcal{F}^{i,\cap}(\bar{\mathbf{Y}}^2)\}\leq\mathcal{F}^{i,\cap}(t\bar{\mathbf{Y}}^1 + (1-t)\bar{\mathbf{Y}}^2),\\
&\mathcal{F}^{i,\cup}(t\bar{\mathbf{Y}}^1 + (1-t)\bar{\mathbf{Y}}^2)\leq\max\{\mathcal{F}^{i,\cup}(\bar{\mathbf{Y}}^1),\mathcal{F}^{i,\cup}(\bar{\mathbf{Y}}^2)\}.
\end{align*}
\normalsize

\end{enumerate}
\end{assumption}
The additional structure imposed by Assumption~\ref{ass:flat} is used for constructing invariant sets for \eqref{eq:sysdyn}, \eqref{eq:nom_dyn} using historical data, which will be clarified in Section \ref{sec:ROLMPC}C.
\begin{remark}
\textit{Assumption~\ref{ass:flat}\eqref{ass:flat_class} is naturally satisfied by flat, simple mechanical systems \cite{murray1997nonlinear} (where the system geometry/kinematics are affected by the control inputs via integrators). The bounding functions $\mathcal{F}^{i,\cap}(\cdot), \mathcal{F}^{i,\cup}(\cdot)$ in Assumption~\ref{ass:flat}\eqref{ass:flatmap_convex} can be constructed by exploiting system constraints, and the required properties can be verified via first \& second order conditions or composition rules for quasiconvex functions \cite{boyd}. If $\mathcal{F}^i(\cdot)$ is both quasiconcave and quasiconvex already, then the bounding functions are simply $\mathcal{F}^{i,\cap}(\cdot)=\mathcal{F}^{i,\cup}(\cdot)=\mathcal{F}^i(\cdot)$.}
\end{remark}

\begin{eg}
\textit{Consider the kinematic bicycle, described by the Euler-discretized dynamics}
{\small{
\begin{align*}
    X_{t+1}&=X_t+\text{dt}(v_t\cos(\theta_t))\\
    Y_{t+1}&=Y_t+\text{dt}(v_t\sin(\theta_t))\\
    \theta_{t+1}&=\theta_t+\text{dt}(\frac{v_t}{L_f}\tan^{-1}(\frac{L_r}{L_f+L_r}\tan(\delta_t)))
\end{align*}}}
\textit{with states $\bar{x}_t=[X_t,Y_t,\theta_t]$ and controls $\bar{u}_t=[v_t, \delta_t$. For the output $\bar{y}_t=[X_t, Y_t]$, the states are reconstructed as $[X_t, Y_t]=\bar{y}_t$, $\theta_t=\tan^{-1}(\frac{[0\ 1](\bar{y}_{t+1}-\bar{y}_t)}{[1\ 0](\bar{y}_{t+1}-\bar{y}_t)})$ and the inputs are reconstructed as $v_t=\frac{\Vert \bar{y}_{t+1}-\bar{y}_t\Vert}{\text{dt}}$, $\delta_t=\tan^{-1}(\frac{L_f+L_r}{L_r}\tan(\frac{L_f(\theta_{t+1}-\theta_t)}{\text{dt}v_t}))$ (which involves $\bar{y}_t,\bar{y}_{t+1},\bar{y}_{t+2}$), and so $R=2$.}

\textit{For Assumption~\ref{ass:flat}\eqref{ass:flat_class}, the states can be reconstructed from two consecutive outputs alone, i.e.,  $\mathcal{F}_x(\bar{y}_t,\bar{y}_{t+1})=[\bar{y}_t, \tan^{-1}(\frac{[0\ 1](\bar{y}_{t+1}-\bar{y}_t)}{[1\ 0](\bar{y}_{t+1}-\bar{y}_t)})]$.}

\textit{For Assumption~\ref{ass:flat}\eqref{ass:flatmap_convex}, the bounding functions corresponding to the states are}
{\small{
\begin{align*}
    &[\mathcal{F}^{1,\cap}(\bar{\mathbf{Y}}_t),\mathcal{F}^{2,\cap}(\bar{\mathbf{Y}}_t)]=[\mathcal{F}^{1,\cup}(\bar{\mathbf{Y}}_t),\mathcal{F}^{2,\cup}(\bar{\mathbf{Y}}_t)]=\bar{y}_t\\
    &\mathcal{F}^{3,\cap}(\bar{\mathbf{Y}}_t)=\mathcal{F}^{3,\cup}(\bar{\mathbf{Y}}_t)=\tan^{-1}(\frac{[0\ 1](\bar{y}_{t+1}-\bar{y}_t)}{[1\ 0](\bar{y}_{t+1}-\bar{y}_t)})
\end{align*}}}
\textit{if $[1\ 0](\bar{y}_{t+1}-\bar{y}_t)>0$ ($\because$ composition of quasi-linear function $\frac{[0\ 1](\bar{y}_{t+1}-\bar{y}_t)}{[1\ 0](\bar{y}_{t+1}-\bar{y}_t)}$ with monotonically increasing function $\tan^{-1}(\cdot)$ is both quasi-convex and quasi-concave). For the bounding functions of the inputs, we have for example,}
{\small{
\begin{align*}
    &\mathcal{F}^{4,\cap}(\bar{\mathbf{Y}}_t)=0, 
    \mathcal{F}^{4,\cup}(\bar{\mathbf{Y}}_t)=\frac{\Vert \bar{y}_{t+1}-\bar{y}_t\Vert}{\text{dt}}\\
    &\mathcal{F}^{5,\cap}(\bar{\mathbf{Y}}_t)=-\frac{\pi}{2}, 
    \mathcal{F}^{5,\cup}(\bar{\mathbf{Y}}_t)=\frac{\pi}{2}
\end{align*}}}
\end{eg}

\subsection{System Constraints}
In this work, we assume that system~\eqref{eq:sysdyn} is subject to state and input constraints given by box sets.
\begin{assumption}\label{ass:box}
The state constraints $\mathcal{X}$ and input constraints $\mathcal{U}$ are given by, 
\begin{align*}
\mathcal{X}=\{x\in\mathbb{R}^n| lb_x\leq x\leq ub_x\},\\ \mathcal{U}=\{u\in\mathbb{R}^m| lb_u\leq u\leq ub_u\}
\end{align*}
for some  $lb_x, ub_x \in\mathbb{R}^n$ and $ lb_u,ub_u\in\mathbb{R}^m$.
\end{assumption}

In Section~\ref{ssec:ConvOPSS}, the box constraints are required to exploit the element-wise bounds on $\mathcal{F}(\cdot)$ given in Assumption~\ref{ass:flat}\eqref{ass:flatmap_convex}. This helps in using system trajectory data to construct lifted outputs that map to nominal state and inputs within constraints.

\subsection{Background: Robust MPC for Nonlinear Systems}
Denote the desired set of states of \eqref{eq:sysdyn} as the goal set $\mathcal{X}_G=\{x\in\mathbb{R}^n| A^g x\leq b^g\}\subset\mathcal{X}$, and suppose that it is control invariant for \eqref{eq:sysdyn} i.e., $\forall x_t\in\mathcal{X}_G, \exists u_t\in\mathcal{U}\Rightarrow  x_{t+1}\in\mathcal{X}_G$.
Define the error system with state $e_t=x_t-\bar x_t$ and dynamics 
\begin{align}\label{eq:error_dyn}
    e_{t+1}&=f_e(e_t, \bar x_t,\bar u_t)+d_t,\ d_t\in D(e_t+\bar x_t, u_t )\nonumber\\
    f_e(e_t,\bar x_t, \bar u_t)&=f(e_t+\bar x_t, \bar u_t+\kappa(e_t))-f(\bar x_t, \bar u_t)
\end{align}
where $u_t$ is given as the sum of $\bar u_t$ and error feedback policy $\kappa(e_t)$. For a fixed policy $\kappa(e_t)$, a robustly positively invariant set $\mathcal{E}$ for the error dynamics~\eqref{eq:error_dyn} satisfies,
\begin{align}\label{eq:e_inv_def}
    &\forall x_t\in\mathcal{X},\forall u_t\in\mathcal{U}:\nonumber \\ 
  &e_t\in\mathcal{E}\Rightarrow e_{t+1}\in\mathcal{E}~\forall d_t\in D(x_t,u_t).
\end{align}
The nominal input $\bar u_t$ is obtained by solving the following finite-horizon optimal control problem for the nominal system,

\begin{equation}\label{eq:OP_RMPC_nom}
	\begin{aligned}
	J_{t}(x_t)= \min\limits_{\mathbf{\bar{u}_t},\mathbf{\bar{x}_t}} \quad & \displaystyle P(\bar x_{t+N|t})+\sum\limits_{k=t}^{t+N-1} c(\bar x_{k|t},\bar u_{k|t}) \\[1ex]
		\text{s.t.}\quad  & \bar x_{k+1|t}=f(\bar x_{k|t},\bar u_{k|t}), \\
    & \bar x_{k|t}\in\bar{\mathcal{X}}, \bar u_{k|t}\in\bar{\mathcal{U}},~\forall k\in \mathbb{I}_{t}^{t+N-1},\\
    & \bar x_{t+N|t}\in\bar{\mathcal{X}}_N,\\
	&   x_t-\bar x_{t|t}\in\mathcal{E}
	\end{aligned}
\end{equation}
where $\mathbf{\bar{u}_t}=[\bar{u}_{t|t},\dots,\bar{u}_{t+N-1|t}],\ \mathbf{\bar{x}_t}=[\bar{x}_{t|t},\dots, \bar{x}_{t+N|t}]$ are decision variables, and $x_t$ is the current state. The notation $\mathbb{I}_{k_1}^{k_2}=\{k_1,...,k_2\}$ is introduced to denote a set of integers between $k_1$ and $k_2$. The optimal solution of \eqref{eq:OP_RMPC_nom} provides the nominal control input as \begin{align}
    \bar{u}_{t}=\pi_{MPC}(x_t)=\bar{u}^\star_{t|t}
\end{align}
and the resulting feedback controller for system \eqref{eq:sysdyn} is
\begin{align}\label{eq:control}
    u_t=\pi(x_t)=\pi_{MPC}(x_t)+\kappa(x_t-x^*_{t|t})
\end{align}
The sets $\bar{\mathcal{X}}\subset\mathcal{X}\ominus \mathcal{E}$ and  $\bar{\mathcal{U}}\subset\mathcal{U}\ominus \kappa(\mathcal{E})$ are tightened nominal state and input constraints, where $\kappa(\mathcal{E})=\{u\in\mathbb{R}^m| \exists e\in\mathcal{E}: \kappa(e)=u\}$ to  ensure that if $\bar{x}_t\in\bar{\mathcal{X}}, \bar{u}_t\in\bar{\mathcal{U}}, e_t\in\mathcal{E}$, then  $x_t\in\mathcal{X}, u_t\in\mathcal{U}$.  The stage cost $c(\cdot)$ is chosen such that
\begin{align*}c(\bar x,\bar{u})=0~ \forall (\bar{x},\bar{u})\in(\mathcal{X}_G\ominus\mathcal{E})\times\bar{\mathcal{U}},\\ c(\bar x, \bar u)>0 ~\forall (\bar x,\bar u)\in\mathbb{R}^{n+m}\backslash(\mathcal{X}_G\ominus\mathcal{E})\times\bar{\mathcal{U}}.
\end{align*}
The terminal cost $P(\cdot)$ and terminal set $\bar{\mathcal{X}}_N$ are designed such that the optimization problem~\eqref{eq:OP_RMPC_nom} has a feasible solution $\forall t\geq 0$ for system \eqref{eq:sysdyn} in closed-loop with the control \eqref{eq:control}, and the state $x_t$ is asymptotically driven to the set $\mathcal{X}_G$. We formally state the elements to be designed for the Robust MPC next.

\subsection{Iterative Design of Robust MPC using Historical Data}
\subsubsection*{Design Elements} We propose an iterative design of the feedback control $u_t=\pi(\cdot)$ in \eqref{eq:OP_RMPC_nom}-\eqref{eq:control} by using historical state-input trajectory data to construct:
\begin{enumerate}[(D1)]
    \item \label{D:einv} the error invariant $\mathcal{E}$ for a fixed error policy $\kappa(\cdot)$ such that \eqref{eq:e_inv_def} is satisfied.
    \item \label{D:constr}the tightened constraints $\bar{\mathcal{X}}\subset\mathcal{X}\ominus\mathcal{E}$ and $\bar{\mathcal{U}}\subset\mathcal{U}\ominus \kappa(\mathcal{E})$ such that $\bar{x}_t\in\bar{\mathcal{X}}$, $\bar{u}_t\in\bar{\mathcal{U}}$, $e_t\in\mathcal{E}$  $\Rightarrow x_t\in\mathcal{X}$, $u_t\in\mathcal{U}$ for system \eqref{eq:sysdyn} with feedback control \eqref{eq:control}.
    \item \label{D:tcost} the terminal constraint $\bar{\mathcal{X}}_N$, and the terminal cost $P(\cdot)$ such that the optimization problem \eqref{eq:OP_RMPC_nom} has a feasible solution for system \eqref{eq:sysdyn} in closed-loop with \eqref{eq:control} $\forall t\geq 0$, and $x_t$ is steered to $\mathcal{X}_G$.
\end{enumerate}
\subsubsection*{Iterative Setup} An iteration as defined as a rollout of system \eqref{eq:sysdyn} starting from a fixed state $x_S$ with some policy $\pi(\cdot)$ such that system state and input remain within constraints, and the system state is asymptotically steered to $\mathcal{X}_G$. Formally, at iteration $j$: \begin{align}\label{eq:iter_spec}
    &\forall t\geq 0:\nonumber\\
    &x^j_{t+1}=f(x^j_t,u^j_t)+d^j_t,~ d^j_t\in D(x^j_t,u^j_t)\quad\nonumber\\
    &x^j_0=x_S,~ x^j_t\in\mathcal{X},~ u^j_t=\pi^j(x_t)\in\mathcal{U},\nonumber\\
    & \text{dist}_{\mathcal{X}_G}(x^j_t)\rightarrow 0
\end{align}
where $\text{dist}_{S}(x)=\min_{y\in S}=||x-y||_2$ is the distance of $x$ from the set $S$, $\pi^j(\cdot)$ is the policy for iteration $j$ and $x_t^j$, $u_t^j$ are the state and input of \eqref{eq:sysdyn} respectively at time $t$. The quantities for the nominal system  \eqref{eq:nom_dyn} are similarly denoted as  $\bar x^j_t, \bar u^j_t, \bar y^j_t$.
 \subsubsection*{Approach Overview} We iteratively synthesize policies $\pi^j(\cdot)$ for iterations $j\geq 1$ using the Robust MPC \eqref{eq:control}, given an initial iteration $0$ satisfying \eqref{eq:iter_spec}. The design elements (D1)-(D3) of the Robust MPC optimization problem at iteration $j$ are constructed using historical nominal trajectory data $\{\{(x^i_t, u^i_t)\}_{t\geq 0}\}_{i=0}^{j-1}$in the following steps:

 \begin{enumerate}
     \item First, we use set-membership techniques to compute outer-approximations of the set \begin{align}\label{eq:calD}
     \mathcal{D}=\bigcup\limits_{(x,u)\in\mathcal{X}\times\mathcal{U}}D(x,u),
 \end{align}which is the set of all disturbance values within constraints,  by using Assumption~\ref{ass:dQC} and trajectory data $\{\{(x^i_t,u^i_t)\}_{t\geq 0}\}_{i=0}^{j-1}$. These outer-approximations get progressively tighter with increasing iterations,  and are used for constructing the error invariant $\mathcal{E}^j$ at iteration $j$ with a fixed error policy $\kappa(\cdot)$ for (D\ref{D:einv}) in Section \ref{ssec:err_inv}.
 \item For the tightened constraints in (D\ref{D:constr}), it suffices to set $\bar{\mathcal{X}}^j=\mathcal{X}\ominus\mathcal{E}^j$, $\bar{\mathcal{U}}^j=\mathcal{U}\ominus\kappa(\mathcal{E}^j)$  to ensure that $\bar{x}^j_t\in\bar{\mathcal{X}}^j, \bar{u}^j_t\in\bar{\mathcal{U}}^j, e^j_t\in\mathcal{E}^j\Rightarrow x^j_t\in\mathcal{X}, u^j_t\in\mathcal{U}$ (by definition of the $\ominus$ operator). However for constructing the terminal set using historical data such that \eqref{eq:OP_RMPC_nom} has a feasible solution $\forall t\geq 0$ (for system \eqref{eq:sysdyn} in closed-loop with \eqref{eq:control}), we impose additional constraints on $\{\{(\bar{x}^i_t,\bar{u}^i_t)\}_{t\geq 0}\}_{i=0}^{j-1}$ in Section~\ref{ssec:constraints} that involve the bounding functions from Assumption \ref{ass:flat}\eqref{ass:flatmap_convex}. 
 \item The terminal set $\bar{\mathcal{X}}_N^j$ is designed by constructing a convex set using data $\{\{[\bar{y}^i_t,..,\bar{y}^i_{t+R-1}]\}_{t\geq 0}\}_{i=0}^{j-1}$, and taking its image under the map $\mathcal{F}_x(\cdot)$ from Assumption \ref{ass:flat}\eqref{ass:flat_class}. We provide a constructive proof in Section \ref{ssec:ConvOPSS} for showing that this set is control invariant: Definition \ref{def:diffFlat} and Assumption \ref{ass:flat}\eqref{ass:flat_class} together guarantee the existence of a control input to keep the state inside the set $\bar{\mathcal{X}}_N^j$, and Assumption \ref{ass:flat}\eqref{ass:flatmap_convex} ensures that this input is within the tightened input constraints $\bar{\mathcal{U}}^j$. The terminal cost is constructed using Barycentric interpolation in the space $[\bar{y}_t,..,\bar{y}_{t+R-1}]$ and verified to be a CLF.
 \end{enumerate}
 
 In Section~\ref{sec:LMPC_analysis}, the resulting trajectories of system \eqref{eq:sysdyn} in closed-loop with $\pi^j(\cdot)$, are shown to satisfy \eqref{eq:iter_spec}, and have non-increasing trajectory costs with increasing iterations.
\section{Robust Output-Lifted Learning MPC}\label{sec:ROLMPC}
In this section, we detail the design of our Robust MPC scheme and its components using our iterative setup.
\subsection{Construction of Error Invariant Set}\label{ssec:err_inv}
For element (D\ref{D:einv}), we first conservatively over-approximate the state-dependent disturbance set $D(\cdot)$ as an i.i.d disturbance $d_t$ with bounded support $\hat{\mathcal{D}}\supset\mathcal{D}$.  Consider the following uncertain system
\begin{align}\label{eq:fake_unc_dyn}
  x_{t+1}&=f(x_t,\bar{u}_t+\kappa(x_t-\bar{x}_t))+d_t\nonumber\\
  \bar{x}_{t+1}&=f(\bar{x}_t,\bar{u}_t)
\end{align}
where $d_t$ is the process noise with support $\hat{\mathcal{D}}$. Define the corresponding error system with error state $e_t=x_t-\bar x_t$ and uncertain dynamics
\begin{align}\label{eq:fake_errdyn}
    e_{t+1}=f_e(e_t,\bar{x}_t,\bar{u}_t)+d_t
\end{align}
where $f_e(\cdot)$ is defined as in \eqref{eq:error_dyn}. Let $\mathcal{E}$ be a Robust  Positive Invariant (RPI) set for \eqref{eq:fake_errdyn}:  
\begin{align*}
&\forall \bar{x}_t\in\bar{\mathcal{X}},\forall \bar{u}_t\in\bar{\mathcal{U}}:\\
&e_t\in\mathcal{E}\Rightarrow e_{t+1}=f_e(e_t,\bar x_t, \bar u_t)+d_t\in\mathcal{E},~\forall d_t\in\hat{\mathcal{D}}. 
\end{align*}
Constructing RPI $\mathcal{E}$ and error policy $\kappa(\cdot)$ for nonlinear systems is difficult in general. However under additional assumptions on $f(\cdot)$ (cf. smoothness, Lipchitz continuity, incremental stabilizability),  it is common in the nonlinear MPC literature to fix a policy $\kappa(e_t)$ and compute $\mathcal{E}$, or compute both jointly \cite{rakovic2005invariant, yu2013tube, jeantube}. In view of this, we make the following assumption to construct the error invariant for the \textit{actual} error dynamics \eqref{eq:error_dyn} in Proposition~\ref{prop:errinv}.
\begin{assumption}\label{ass:einv}
Given disturbance support $\hat{\mathcal{D}}\supset\mathcal{D}$ and a known, fixed linear policy $\kappa(e_t)=Ke_t$, the RPI set $\mathcal{E}$ for \eqref{eq:fake_errdyn} can be computed such that
\begin{align*}
&\forall \bar{x}_t\in\bar{\mathcal{X}},\forall \bar{u}_t\in\bar{\mathcal{U}}:\\
&e_t\in\mathcal{E}\Rightarrow e_{t+1}=f_e(e_t,\bar x_t, \bar u_t)+d_t\in\mathcal{E},~\forall d_t\in\hat{\mathcal{D}}. 
\end{align*}
\end{assumption}
\begin{proposition}\label{prop:errinv}
Let $\mathcal{E}$ by a Robust Positive Invariant set for \eqref{eq:fake_errdyn} with $\kappa(e_t)=Ke_t$ for some disturbance support $\hat{\mathcal{D}}\supset\mathcal{D}$. Then for the \textit{actual} error system given by \eqref{eq:error_dyn} with error state $e_t=x_t-\bar x_t$, we have
\begin{align*}
&\forall \bar{x}_t\in\mathcal{X}\ominus\mathcal{E},\forall \bar{u}_t\in\mathcal{U}\ominus K\mathcal{E}:\\
&e_t\in\mathcal{E}\Rightarrow e_{t+1}=f_e(e_t,\bar x_t, \bar u_t)+d_t\in\mathcal{E},\forall d_t\in D(e_t+\bar{x}_t,u_t). 
\end{align*}
\end{proposition}
\begin{proof}
Since $\mathcal{E}$ is RPI for \eqref{eq:fake_errdyn} with $\kappa(e_t)=Ke_t$ and disturbance support $\hat{\mathcal{D}}$, we have  
\begin{align*}
&\forall \bar{x}_t\in\mathcal{X}\ominus \mathcal{E},\forall \bar{u}_t\in\mathcal{U}\ominus K\mathcal{E}:\\
&e_t\in\mathcal{E}\Rightarrow e_{t+1}=f_e(e_t,\bar x_t, \bar u_t)+d_t\in\mathcal{E},~\forall d_t\in\hat{\mathcal{D}}. 
\end{align*}
Since $x_t=\bar{x}_t+e_t\in\mathcal{X}$ and $u_t=\bar{u}_t+Ke_t\in\mathcal{U}$ for any $e_t\in\mathcal{E}$, $\bar{x}_t\in\mathcal{X}\ominus\mathcal{E}$, $\bar{u}_t\in\mathcal{U}\ominus K\mathcal{E}$, we have $d_t\in D(e_t+\bar{x}_t,u_t)\subset\mathcal{D}$ from \eqref{eq:calD}. The desired result then follows because $\mathcal{D}\subset\hat{\mathcal{D}}$.
\end{proof}
 The final step towards obtaining $\mathcal{E}$ for (D\ref{D:einv}) is to construct an outer-approximation of $\mathcal{D}$. We use state-input trajectory data collected from our iterative setup using set membership techniques to approximate the \textit{graph} of $D(\cdot)$, defined below. 

\begin{definition}[Graph]\label{def:graph}
The graph of the set-valued map $D: \mathbb{R}^{n+m}\rightarrow 2^{\mathbb{R}^n}$ is defined as the set 
\begin{align}\label{graphdef}
    G(D)=\{(x,u,z)\in\mathbb{R}^{n+m+n}: \ (x,u)\in\mathbb{R}^{n+m},z\in D(x,u)\}.
\end{align}
\end{definition}

In the next proposition we establish an approximation of the graph $G(D)$ using the trajectory data $\{\{(x^i_t,u^i_t)\}_{t\geq0}\}_{i=0}^{j-1}$ and Assumption \ref{ass:dQC}.
\begin{proposition}\label{prop:qc}
Given data $\{x^i_t, u^i_t, x^i_{t+1}\}$ for system \eqref{eq:sysdyn}, define $d^i_t=x^i_{t+1}-f(x^i_t,u^i_t)$,  $q^i_t=(x^i_t,u^i_t)$, where $d^i_t\in D(x^i_t,u^i_t)$ from Assumption~\ref{ass:dQC} and consider the set
\begin{align}\label{eq:E_approx_jt}
    \mathcal{G}^i_{t}=\left\{(q,z)\middle|\ \begin{aligned}{\small{ \begin{bmatrix}1\\q-q^i_t\\z-d^i_{t}\end{bmatrix}^\top Q^{(j)} \begin{bmatrix}1\\q-q^i_t\\z-d^i_{t}\end{bmatrix}\geq 0 }},~\forall Q^{(j)}\in\boldsymbol{\mathcal{Q}}\end{aligned}\right\} .
\end{align}
Then the graph of $D(\cdot)$ satisfies,
\begin{align}\label{eq:E_inter}
    G(D)\subseteq \bigcap_{i=0}^{j-1}\bigcap_{t\geq0}\mathcal{G}^i_t.
\end{align}
\end{proposition}
\begin{proof}\textit{Proof in the appendix}
\end{proof}
Now we present an approach to use the bound \eqref{eq:E_inter} on $G(D)$ proposition to construct $\hat{\mathcal{D}}\supset\mathcal{D}$ using semi-definite programming. We use the S-procedure to construct an ellipsoidal outer-approximation of the set in the RHS of \eqref{eq:E_inter}.

\begin{theorem}\label{thm:D_cnstrct}[\textbf{Support Construction}]
Given trajectories $\{\{(x^i_t,u^i_t)\}_{t\geq0}\}_{i=0}^{j-1}$ of system \eqref{eq:sysdyn} satisyfing Assumption~\ref{ass:dQC}, let $\lambda^{j\star}>0, S^{j\star}\in\mathbb{S}^{n}_{++}, c^{j\star}\in\mathbb{R}^n$ be the optimal solution of the SDP:
    \small
    \begin{align}\label{eq:LMI0}
    &\min_{\substack{S^j, c^j, \lambda^j,B_x, B_u, \tau^{i,l}_t\\
\forall t\in\mathbb{I}_0^{T_{max}},~\forall l\in\mathbb{I}_1^{|\boldsymbol{\mathcal{Q}}|},~\forall i\in\mathbb{I}_0^{j-1}}} \text{tr}(S^j)\nonumber\\
    &~~\text{s.t}\left[
\begin{array}{cccc}
    \multicolumn{3}{c}{\multirow{3}{*}{$\mathbf{M}^j_{11}$}} & -c^{j\top} \\
    \multicolumn{3}{c}{} & 0_n  \\
    \multicolumn{3}{c}{} &  \lambda^jI_n\\
     & \star & & S^j \\
\end{array} \right]\succeq 0,\nonumber\\
&~~~~~~\mathbf{M}^j_{11}=\lambda^je_{1}e_{1}^\top-M_{x}-M_u-\sum\limits_{i=0}^{j-1}\sum\limits_{t=0}^{T_{max}}\sum\limits_{l=1}^{|\boldsymbol{\mathcal{Q}}|}\tau_{t}^{i,l}M^{i,l}_t,\nonumber\\
&~~~~~~M_{x}=\begin{bmatrix}-lb_x^\top (B_x+B^\top_x) ub_x&(ub_x-lb_x)^\top B_x & O_n\\ &-B_x-B^\top_x & O_n\\ \star& &O_n\end{bmatrix},\nonumber\\
&~~~~~~M_{u}=\begin{bmatrix}-lb_u^\top (B_u+B^\top_u) ub_u&(ub_u-lb_u)^\top B_u & O_n\\ &-B_u-B^\top_u & O_n\\ \star& &O_n\end{bmatrix},\nonumber\\
&~~~~~~M^{i,l}_{t}=(I_{2n+1}-\begin{bmatrix}
    0\\q_t^i\\d_t^i
    \end{bmatrix}e_1^\top)^\top Q^{(l)}(I_{2n+1}-\begin{bmatrix}
    0\\q_t^i\\d_t^i
    \end{bmatrix}e_1^\top),\nonumber\\
 &~~~~~~S^j\succ 0, \lambda^j>0, B_x\geq 0, B_u\geq 0, \tau^{i,l}_t\geq 0,\nonumber\\
 &~~~~~~~\forall t\in\mathbb{I}_0^{T_{max}}, l\in\mathbb{I}_1^{|\boldsymbol{\mathcal{Q}}|}, i\in\mathbb{I}_0^{j-1} \nonumber\\
 &~~~~\text{if } j>1:\nonumber\\
 &~~~~~~\begin{bmatrix}1 & -\tilde{c}^{j-1\top}\\
 \star & \tilde{S}^{j-1}
 \end{bmatrix}\geq \begin{bmatrix}
 \lambda^j & -c^{j\top}\\
 \star & S^j
 \end{bmatrix}
    \end{align}
    \normalsize
   Then for $\tilde{S}^{j}=\frac{1}{\lambda^{\star j}}S^{\star j}, \tilde{c}^j=\frac{1}{\lambda^{\star j}}c^{j\star}$, the ellipsoidal set $\hat{\mathcal{D}}^{j}=\{d: (d-\tilde{c}^{j})^\top (\tilde{S}^{j})^{-1}(d-\tilde{c}^{j})\leq 1\}\supset \mathcal{D}$. Moreover, feasibility of \eqref{eq:LMI0} for $j=1$, implies feasibility $\forall j>1$ and $\hat{\mathcal{D}}^{j}\subset\hat{\mathcal{D}}^{j-1}$.

\end{theorem}
\begin{proof}Take the Schur complement of the first LMI in \eqref{eq:LMI0} and multiply from both sides by $[1\ q^\top\ z^\top]\ \text{and}\ [1\ q^\top\ z^\top]^\top$ (where $q=(x,u)$) to get
\begin{align*}
    \lambda^{j\star}&\left(1-(z-\frac{c^{j\star}}{\lambda^{j\star}})^\top\left(\frac{S^{j\star}}{\lambda^{j\star}}\right)^{-1}(z-\frac{c^{j\star}}{\lambda^{j\star}})\right)\succ\\
    &(lb_x-x)^\top  (B^\star_x+B^{\star\top}_x)(x-ub_x)+\\
    &(lb_u-u)^\top  (B^\star_u+B^{\star\top}_u)(u-ub_u)+\\
    &\sum\limits_{i=0}^{j-1}\sum\limits_{t=0}^{T_{max}}\sum\limits_{l=1}^{|\boldsymbol{\mathcal{Q}}|}\tau_t^{i,l}\begin{bmatrix}1\\x-x^i_t\\z-d^i_{t}\end{bmatrix}^\top Q^{(l)} \begin{bmatrix}1\\x-x^i_t\\z-d^i_{t}\end{bmatrix}.
\end{align*}
Then $\forall z\in D(x,u)$, $x\in\mathcal{X}, u\in\mathcal{U}$, the three terms on the RHS are positive. For $\lambda^{j\star}>0$ we have $(z-\tilde{c}^j)^\top(\tilde{S}^{j})^{-1}(z-\tilde{c}^j)\leq 1$,  which proves $\hat{\mathcal{D}}^j\supset\mathcal{D}$.
For $j>1$, multiplying the last LMI in \eqref{eq:LMI0} from both sides by $[1\ z^\top], [1\ z^\top]^\top$ gives $1-(z-\tilde{c}^{j-1})^\top (\tilde{S}^{j-1})^{-1}(z-\tilde{c}^{j-1})\geq \lambda^{j\star}(1-(z-\tilde{c}^{j})^\top (\tilde{S}^{j})^{-1}(z-\tilde{c}^j)$, which implies $\hat{\mathcal{D}}^{j}\subset\hat{\mathcal{D}}^{j-1}$. Feasibility of \eqref{eq:LMI0} for $j>1$ can be verified by noticing that the solution for $j=1$ is also feasible for $j>1$, with $\tau^{i,l}_t=0 ~\forall i\geq1, i<j$.
\end{proof}

\begin{assumption}\label{ass:E_init}
We are provided with a state-input trajectory $\{(x^0_t,u^0_t)\}_{t\geq0}$ of system \eqref{eq:sysdyn} for iteration $0$ for which SDP \eqref{eq:LMI0} is feasible for construction of $\mathcal{E}^{1}$ using Assumption \ref{ass:einv}.
\end{assumption}

By Assumption \ref{ass:E_init}, the trajectory data of iteration $j=0$ can be used to construct $\hat{\mathcal{D}}^1$ by solving \eqref{eq:LMI0}, and consequently by Assumption \ref{ass:einv}, the error invariant set $\mathcal{E}^1$ for a fixed policy $\kappa(e_t)=Ke_t$. For subsequent iterations, $\hat{\mathcal{D}}^j$ is given by solving \eqref{eq:LMI0} with an additional constraint for enforcing $\hat{\mathcal{D}}^j\subset \hat{\mathcal{D}}^{j-1}$, and RPI $\mathcal{E}^j$ is constructed with  $\kappa(e_t)=Ke_t$. By  Proposition~\ref{prop:errinv}, the sets $\mathcal{E}^i~ \forall i=1,\dots,j$ are positively invariant for the error dynamics \eqref{eq:error_dyn}. Since $\hat{\mathcal{D}}^j\subset\cdot\cdot\subset\hat{\mathcal{D}}^{1}$, we have 
    $\mathcal{E}^j\subseteq\cdot\cdot\subseteq\mathcal{E}^1$.

\subsection{Tightened State and Input Constraints}\label{ssec:constraints}
Given the error invariant $\mathcal{E}^{j}$ and policy $\kappa(e_t)=Ke_t$, consider the tightened state and input constraints $\hat{\mathcal{X}}^{j}=\mathcal{X}\ominus\text{Box}(\mathcal{E}^{j})=\{x| lb^j_x\leq x\leq ub^j_x\}$ and $\hat{\mathcal{U}}^j=\mathcal{U}\ominus \text{Box}(K\mathcal{E}^{j})=\{u| lb^j_u\leq u\leq ub^j_u\}$ where $\text{Box}(S)$ denotes smallest box set that contains $S$. We tighten $\hat{\mathcal{X}}^{j}\times\hat{\mathcal{U}}^{j}$ further such that for any $(\bar{x},\bar{u})\in\hat{\mathcal{X}}^{j}\times \hat{\mathcal{U}}^{j}$ with $\mathcal{F}(\bar{\mathbf{Y}})=(\bar{x},\bar{u})$, we also have $[\mathcal{F}^{\cap}(\bar{\mathbf{Y}}),\mathcal{F}^{\cup}(\bar{\mathbf{Y}})]\subset \hat{\mathcal{X}}^{j}\times \hat{\mathcal{U}}^{j}$. Since $\hat{\mathcal{X}}^{j}\times \hat{\mathcal{U}}^{j}$ are box constraints, this tightening can be expressed as
\begin{align*}
\bar{\mathcal{S}}^j_x&=\left\{
\bar{x}\in\hat{\mathcal{X}}^{j}\middle\vert
\begin{aligned}
&\exists\bar{\mathbf{y}}\in\mathbb{R}^{m\times R}, \bar{x}=\mathcal{F}_x(\bar{\mathbf{y}})\\
&lb^j_x\leq\mathcal{F}_x^{\cap}(\bar{\mathbf{y}}), \mathcal{F}_x^{\cup}(\bar{\mathbf{y}})\leq ub^j_x
\end{aligned}
\right\}\\
\bar{\mathcal{S}}^j_u&=\left\{
\bar{u}\in\hat{\mathcal{U}}^{j}\middle\vert
\begin{aligned}
&\exists\bar{\mathbf{Y}}\in\mathbb{R}^{m\times R+1}, \bar{u}=\mathcal{F}_u(\bar{\mathbf{Y}})\\
&lb^j_u\leq\mathcal{F}_u^{\cap}(\bar{\mathbf{Y}}), \mathcal{F}_u^{\cup}(\bar{\mathbf{Y}})\leq ub^j_u
\end{aligned}
\right\}
\end{align*}
where $\mathcal{F}^{\cup}_x(\cdot)=[\mathcal{F}^{1,\cup}(\cdot),..,\mathcal{F}^{n,\cup}(\cdot)],\ \mathcal{F}^{\cup}_u(\cdot)=[\mathcal{F}^{n+1,\cup}(\cdot),..,\mathcal{F}^{n+m,\cup}(\cdot)]$ and $\mathcal{F}^\cap_x(\cdot), \mathcal{F}^\cap_u(\cdot)$ are defined similarly.
We obtain box sets $\bar{\mathcal{X}}^{j}$, $\bar{\mathcal{U}}^{j}$ for iteration $j$ recursively, by using the box sets $\bar{\mathcal{X}}^{j-1}$, $\bar{\mathcal{U}}^{j-1}$ from iteration $j-1$ to solve the following nonlinear program,
\begin{align}\label{eq:tight_const_NLP}
\max_{\substack{\alpha^j_x\geq 1,\alpha^j_u\geq 1,v^j_x,v^j_u}}\quad&\alpha^j_x+\alpha^j_u\nonumber\\
\ \text{s.t}\quad &\ \bar{\mathcal{X}}^{j-1}\subseteq\alpha^j_x\bar{\mathcal{X}}^{j-1}+v^j_x\subseteq \bar{\mathcal{S}}^j_{x},\nonumber\\
 &\ \bar{\mathcal{U}}^{j-1}\subseteq \alpha^j_u\bar{\mathcal{U}}^{j-1}+v^j_u\subseteq \bar{\mathcal{S}}^j_{u}
\end{align}
\normalsize
to consequently define the tightened state and input constraints for our Robust MPC design as
\begin{align}\label{eq:tight_cnstrnts}
\bar{\mathcal{X}}^{j}=\alpha^{j\star}_x\bar{\mathcal{X}}^{j-1}+v^{j\star}_x,\nonumber\\ \bar{\mathcal{U}}^{j}=\alpha^{j\star}_u\bar{\mathcal{U}}^{j-1}+v^{j\star}_u,
\end{align}
as illustrated in Figure~\ref{fig:xbar}. 
The recursion \eqref{eq:tight_cnstrnts} is initialised by solving \eqref{eq:tight_const_NLP} with relaxed constraints $\alpha^1_x\bar{\mathcal{X}}^{0}+v^1_x\subseteq\bar{\mathcal{S}}^1_x$, $\alpha^1_u\bar{\mathcal{U}}^{0}+v^1_u\subseteq\bar{\mathcal{S}}^1_u$, $\alpha^j_x\geq 0$,$\alpha^j_u\geq 0$. 
\begin{figure}[h]
    \centering
    \includegraphics[scale=0.32]{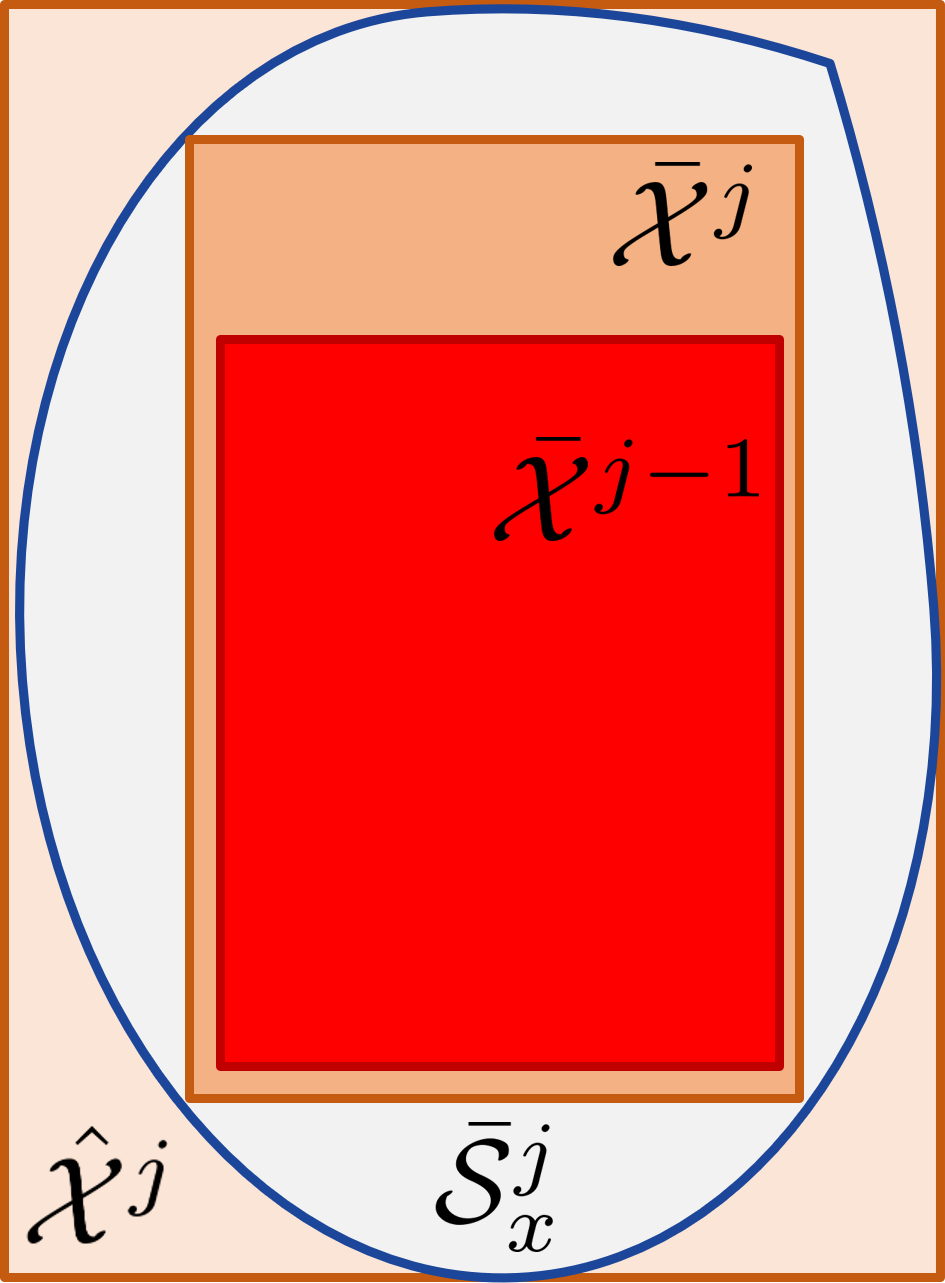}
    \caption{{\small{Construction of tightened state constraints $\bar{\mathcal{X}}^j$ via \eqref{eq:tight_const_NLP},\eqref{eq:tight_cnstrnts}. The set $\bar{\mathcal{S}}^j_x\subset \hat{\mathcal{X}}^j$ is constructed so that $lb^j_x\leq \mathcal{F}^{\cap}_x(\bar{\mathbf{Y}})\leq \mathcal{F}_x(\bar{\mathbf{Y}})=\bar{x}\leq\mathcal{F}^{\cup}_x(\bar{\mathbf{Y}})\leq ub^j_x$, which is required for property 3) in proposition~\ref{prop:tight_cons_NLP}.}}}
    \label{fig:xbar}
\end{figure}
\begin{remark}
\textit{The box sets' inclusion constraint $\bar{\mathcal{X}}^{j-1}\subseteq\alpha^j_x\bar{\mathcal{X}}^{j-1}+v^j_x$ in \eqref{eq:tight_const_NLP} is enforced using the diagonally opposite vertices via $2\cdot n$ inequalities. To enforce $\alpha^j_x\bar{\mathcal{X}}^{j-1}+v^j_x\subseteq \bar{\mathcal{S}}^j_{x}$, the $n-1$ dimensional facets of  $\alpha^j_x\bar{\mathcal{X}}^{j-1}+v^j_x$ are gridded and each grid point is constrained to lie within $ \bar{\mathcal{S}}^j_{x}$ (which is non-convex but simply connected\footnote{Shown by exploiting the continuity, surjectivity of $\mathcal{F}(\cdot)$ and convexity of the set $\{\bar{\mathbf{y}}| lb^j_x\leq \mathcal{F}^\cup_x(\bar{\mathbf{y}}), \mathcal{F}^\cap_x(\bar{\mathbf{y}})\leq ub^j_x\}$ \cite[Chapter 9]{munkres2000topology}.}, i.e., has no holes). The constraints involving $\bar{\mathcal{U}}^{j-1}$ are enforced similarly.}
\end{remark}
\begin{proposition}\label{prop:tight_cons_NLP} Given Assumption~\ref{ass:flat}, the tightened constraints \eqref{eq:tight_cnstrnts} and the nonlinear program \eqref{eq:tight_const_NLP} used for their construction satisfy the following properties:
\begin{enumerate}
    \item If problem \eqref{eq:tight_const_NLP} is feasible at iteration $j=1$, then it remains feasible for $\forall j\geq 1$ with $\bar{\mathcal{X}}^{j}\supseteq \bar{\mathcal{X}}^{j-1}$, $\bar{\mathcal{U}}^{j}\supseteq \bar{\mathcal{U}}^{j-1}$.
    \item $\bar{x}_t\in\bar{\mathcal{X}}^{j}, e_t\in\mathcal{E}^j, \bar{u}_t\in\bar{\mathcal{U}}^{j}\Rightarrow x_t\in\mathcal{X}, u_t\in\mathcal{U}$.
    \item Let $\{\bar{\mathbf{Y}}^1,\dots,\bar{\mathbf{Y}}^p\}$ be a set of lifted-outputs such that $\mathcal{F}(\bar{\mathbf{Y}}^k)\in\bar{\mathcal{X}}^j\times\bar{\mathcal{U}}^j~\forall k=1,\dots, p$. Then $\mathcal{F}(\bar{\mathbf{Y}})\in \bar{\mathcal{X}}^j \times \bar{\mathcal{U}}^j$ for any $\bar{\mathbf{Y}}\in\textrm{conv}(\{\bar{\mathbf{Y}}^1,\dots,\bar{\mathbf{Y}}^p\})$.
\end{enumerate}
\end{proposition}
\begin{proof} \textit{Proof in the appendix}
\end{proof}
 Property 1) ensures that recursion \eqref{eq:tight_cnstrnts} succeeds for iterations $j>1$. Property 2) proves that the tightened constraints satisfy (D\ref{D:constr}). Property 3) shows that given any set of historical lifted outputs mapping to state-input pairs within $\bar{\mathcal{X}}^j\times\bar{\mathcal{U}}^j$, \textit{any} convex combination of these historical lifted outputs also maps to state-input pairs within $\bar{\mathcal{X}}^j\times\bar{\mathcal{U}}^j$. This is proved using the bounding functions from Assumption \ref{ass:flat}\eqref{ass:flatmap_convex}, which are incorporated into the tightened constraints via $\bar{\mathcal{S}}^j_{x}, \bar{\mathcal{S}}^j_{u}$. In the next sub-section we construct a control invariant set using historical lifted output data that map to nominal state-input pairs within the tightened constraints.
\subsection{Terminal Set and Terminal Cost}\label{ssec:ConvOPSS}
In this section, we use historical trajectory data to iteratively construct the terminal set and terminal cost for our Robust MPC. Let $\{(x^j_t,u^j_t)\}_{t\geq 0}$ be a system trajectory satisfying the properties in \eqref{eq:iter_spec}. Then we recursively define the nominal \textit{Convex Output Safe Set} $\bar{\mathcal{CS}}^{j}_{\mathbf{y}}$ in three steps:
\begin{itemize}
    \item Step 1) Construct nominal trajectory data $\{(\bar{x}^{j}_t,\bar{u}^{j}_t)\}_{t\geq 0}$ satisfying:
    \begin{subequations}\label{eq:nom_spec}
    \begin{align}
    &x^{j}_t-\bar{x}^{j}_t\in\mathcal{E}^{j+1},\  \bar{x}^{j}_t\in\bar{\mathcal{X}}^{j+1}\label{eq:nom_a}\\
    &\bar{u}^{j}_t=u^{j}_t-K(x^{j}_t-\bar{x}^{j}_t)\in\bar{\mathcal{U}}^{j+1}\label{eq:nom_b}\\
    &\bar{x}^{j}_{t+1}=f(\bar{x}^{j}_t,\bar{u}^{j}_t)~~\forall t\geq 0\label{eq:nom_c}
    \end{align}
    \end{subequations}

    \item Step 2) Define $\bar{\mathbf{y}}^{j}_t=[\bar{y}^{j}_t,\dots,\bar{y}^{j}_{t+R-1}]\in\mathbb{R}^{m\times R}$ $\forall t\geq 0$ where $\bar{y}^{j}_t=h(\bar{x}^{j}_t)$.
    \item Step 3) Define the set $\bar{\mathcal{CS}}^{j}_{\mathbf{y}}$ as
    \begin{align}\label{eq:CSy_def}
    \bar{\mathcal{CS}}^{j}_{\mathbf{y}}=\text{conv}(\bar{\mathcal{CS}}^{j-1}_{\mathbf{y}}\bigcup_{t\geq 0}\bar{\mathbf{y}}^{j}_t).
    \end{align}
\end{itemize}


Note that each $\bar{\mathbf{y}}^{j}_t$ constructed in Step 2) uniquely identifies the nominal state $\bar{x}^{j}_{t}$ via the map \eqref{eq:flatclass_x}, $\mathcal{F}_x(\bar{\mathbf{y}}^{j}_t)=\bar{x}^{j}_t$ and similarly $\mathcal{F}_u(\bar{\mathbf{Y}}^{j}_t)=\bar{u}^{j}_t$ where $\bar{\mathbf{Y}}^{j}_t=[\bar{y}^{j}_{t},\dots,\bar{y}^{j}_{t+R}]\in\mathbb{R}^{m\times R+1}$ is the lifted output at time $t$ and iteration $j$.

Define  the forward-time shift $\delta(\cdot,\cdot)$ dynamics on $\bar{\mathcal{CS}}_{\mathbf{y}}^{j}$ as
\begin{align} \label{eq:time_shift}
   \bar{\mathbf{y}}_{t+1}&=[\bar{y}_{t+1},\bar{y}_{t+2},\dots,\bar{y}_{t+R}]\nonumber\\&=\delta( [\bar{y}_{t}, \bar{y}_{t+1},\dots,\bar{y}_{t+R-1}], \bar{y}_{t+R})\nonumber \\
   &=\delta(\bar{\mathbf{y}}_t,\bar{y}_{t+R})
\end{align}\normalsize
We now show that the Convex Safe Set  $\bar{\mathcal{CS}}_{\mathbf{y}}^{j}$ is in fact, control invariant for \eqref{eq:time_shift} in the following proposition and correspond to nominal states and inputs within constraints (as depicted in Figure~\ref{fig:prop_exp}).
\begin{figure}[h]
    \centering
    \includegraphics[scale=0.3]{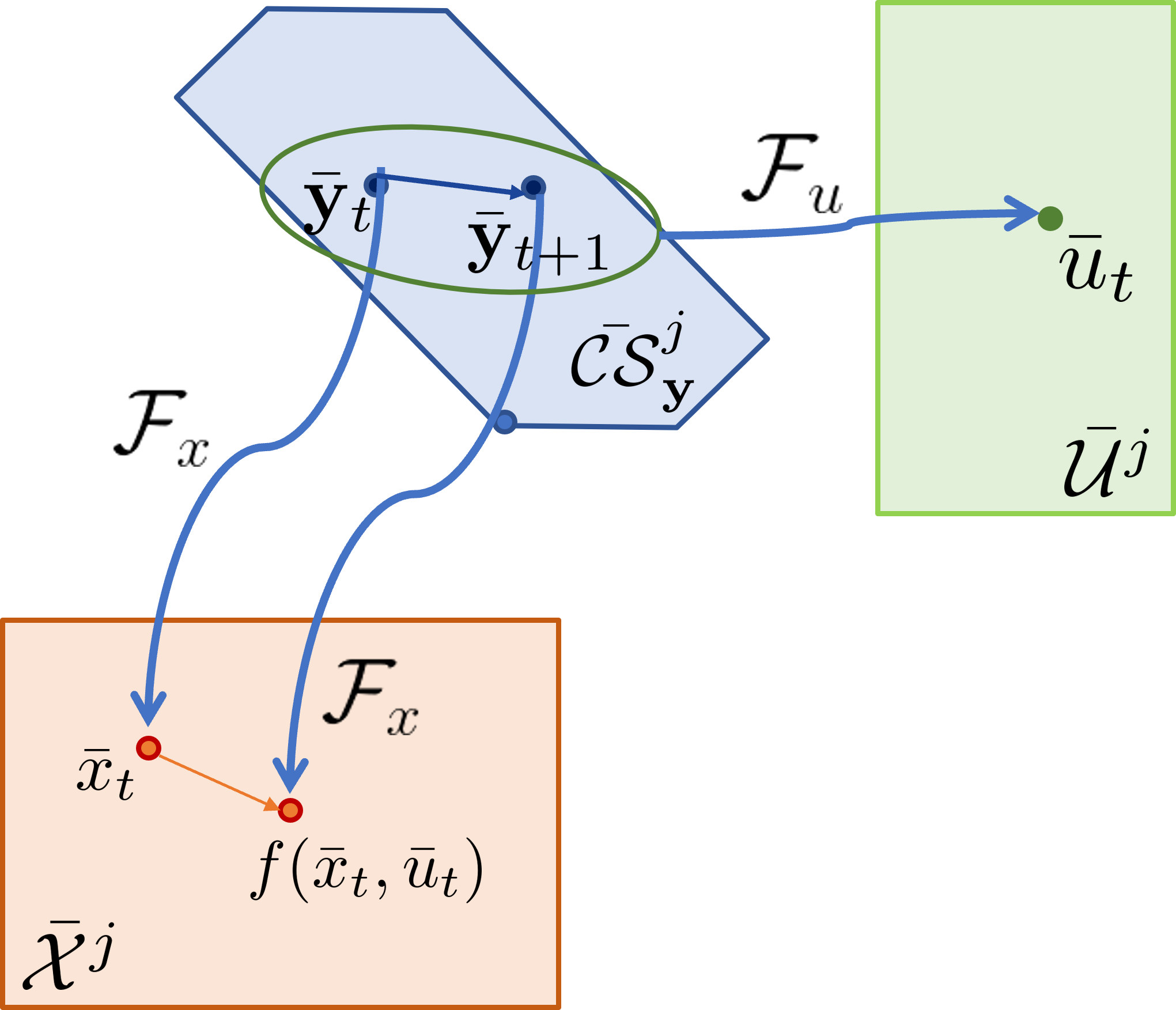}
    \caption{{\small{Illustration of the claim in Proposition \ref{prop:CS_CI}.}}}
    \label{fig:prop_exp}
\end{figure}
\begin{proposition}\label{prop:CS_CI}
Under Assumptions~\ref{ass:flat} and \ref{ass:box}, the set $\bar{\mathcal{CS}}_{\mathbf{y}}^{j}$ \eqref{eq:CSy_def} is control invariant for the forward-time shift dynamics \eqref{eq:time_shift}, i.e.,
\begin{align}\label{eq:imp_2_CI}
 \forall\bar{\mathbf{y}}_t\in\bar{\mathcal{CS}}_{\mathbf{y}}^{j},\exists \bar{y}_{t+R}\ :\ \bar{\mathbf{y}}_{t+1}=\delta(\bar{\mathbf{y}}_t,\bar{y}_{t+R})\in\bar{\mathcal{CS}}_{\mathbf{y}}^{j},    
\end{align} with
\begin{align*}
    \bar{x}_t=\mathcal{F}_x(\bar{\mathbf{y}}_t)\in\bar{\mathcal{X}}^j,&\ \bar{u}_t=\mathcal{F}_u(\bar{\mathbf{y}}_t,\bar{y}_{t+R})\in\bar{\mathcal{U}}^j,\\
    \mathcal{F}_x(\bar{\mathbf{y}}_{t+1})&=f(\bar{x}_t,\bar{u}_t)\in\bar{\mathcal{X}}^j.
\end{align*}
\end{proposition}

\begin{proof}\textit{Proof in the appendix}\end{proof}

Given trajectory data $\{\{(x^i_t,u^i_t)\}_{t\geq 0}\}_{i=0}^{j-1}$, the terminal set for our Robust MPC problem at iteration $j$ is defined as 
\begin{align}\label{eq:term_set}
    \bar{\mathcal{X}}^j_N:=\{\bar{x}\ | \exists\bar{\mathbf{y}}\in \bar{\mathcal{CS}}_{\mathbf{y}}^{j-1},\mathcal{F}_x(\bar{\mathbf{y}})=\bar{x}\}.
\end{align}
\begin{remark}\textit{Note that this set is constructed without any local linear approximations of the nominal dynamics \eqref{eq:nom_dyn}, and uses the full nonlinear dynamics \textit{implicitly} via the lifted output data $\{\{\bar{\mathbf{y}}^i_t\}_{t\geq 0}\}_{i=0}^{j-1}$ and map $\mathcal{F}(\cdot)$. }\end{remark}
Now we proceed to construct a terminal cost function which approximates the optimal cost-to-go from a state using lifted outputs from previous iterations. Construct box sets $\bar{\mathcal{X}}_G=\{x| \tilde{lb}^g_x\leq x\leq \tilde{ub}^g_x\}\subset\mathcal{X}_G\ominus\mathcal{E}^1$, and $\mathcal{U}_{nom}=\{u| \tilde{lb}^{nom}_u\leq u\leq \tilde{ub}^{nom}_u\}\subset \mathcal{U}\ominus K\mathcal{E}^1$ . Define the convex set \begin{align}\mathcal{Y}_G=\left\{\mathbf{Y}\in\mathbb{R}^{m\times R+1}\middle\vert\begin{aligned} &\mathbf{Y}=[\mathbf{y}, y],\\
&\tilde{lb}^g_x\leq\mathcal{F}_x^{\cap}(\mathbf{y}),\mathcal{F}_x^{\cup}(\mathbf{y})\leq \tilde{ub}^g_x, \\ &\tilde{lb}^{nom}_u\leq\mathcal{F}_u^{\cap}(\mathbf{Y}), \mathcal{F}_u^{\cup}(\mathbf{Y})\leq \tilde{ub}^{nom}_u \end{aligned}\right\},\end{align} and observe that $\bar{\mathbf{Y}}=[\bar{\mathbf{y}},y]\in\mathcal{Y}_G\Rightarrow \bar{x}=\mathcal{F}_x(\bar{\mathbf{y}})\in\mathcal{X}_G\ominus\mathcal{E}^1,  \bar{u}=\mathcal{F}_u(\bar{\mathbf{Y}})\in\mathcal{U}_{nom}$ by Assumption \ref{ass:flat}\eqref{ass:flatmap_convex}. For some iteration $i$ and time $t$, define the quantity
\begin{equation}\label{eq:ctg}
    \mathcal{C}^i_t=\sum\limits_{k\geq t} c(\bar{\mathbf{Y}}^i_k)
\end{equation}\normalsize
 where the function $c(\cdot)$ is convex, continuous and satisfies
 \begin{align}\label{stage_cost}
         c(\bar{\mathbf{Y}})=0~ \forall \bar{\mathbf{Y}}\in\mathcal{Y}_G,\  c(\bar{\mathbf{Y}})\succ 0 \ \forall \bar{\mathbf{Y}}\in\mathbb{R}^{m\times R+1}\backslash \mathcal{Y}_G.
 \end{align}\normalsize

For iteration $j$, we use \eqref{eq:ctg} to construct a cost function on the convex safe set $\bar{\mathcal{CS}}_{\mathbf{y}}^{j-1}$ using Barycentric interpolation \cite{jones2010polytopic} with tuples $(\bar{\mathbf{y}}^i_t, \mathcal{C}^i_t), \forall \bar{\mathbf{y}}^i_t\in\bar{\mathcal{CS}}_{\mathbf{y}}^{j-1}$.
\begin{equation}\label{eq:conv_cf}
\begin{aligned}
    \bar{Q}^{j-1}(\bar{\mathbf{y}})=\min\limits_{\substack{\lambda^i_k\in[0,1],\ \forall t\geq 0\\\forall i\in\mathbb{I}_0^{j-1}}} \quad & \sum\limits_{i=0}^{j-1}\sum\limits_{k\geq 0}\lambda^i_k\mathcal{C}^i_k\\[1ex]
		\text{s.t.} ~~\quad & \sum\limits_{i=0}^{j-1}\sum\limits_{k\geq 0}\lambda^i_k\bar{\mathbf{y}}^i_k=\bar{\mathbf{y}},\\
		& \sum\limits_{i=0}^{j-1}\sum\limits_{k\geq 0}\lambda^i_k=1
\end{aligned}
\end{equation}\normalsize

We make the following assumption to initialize our constructions of $\bar{\mathcal{CS}}^{j}_{\mathbf{y}}$ and $\bar{Q}^{j}(\cdot)$.
\begin{assumption}\label{ass:CS_init}
We are provided with a nominal state-input trajectory $\{(\bar{x}^0_t,\bar{u}^0_t)\}_{t\geq0}$ of system \eqref{eq:nom_dyn} for iteration $0$ satisfying \eqref{eq:nom_spec} to define $\bar{\mathcal{CS}}^{0}_{\mathbf{y}}=\text{conv}(\bigcup_{t\geq 0}\bar{\mathbf{y}}^{0}_t)$,  and with $\mathcal{C}^0_0<\infty$ to define $\bar{Q}^0(\cdot)$.
\end{assumption}

 The following proposition identifies CLF-like characteristics of the function \eqref{eq:conv_cf} on the set $\bar{\mathcal{CS}}_{\mathbf{y}}^{j}$ which we will use for convergence analysis in Section~\ref{sec:LMPC_analysis}.
\begin{proposition}\label{prop:Q_CLF}
Given Assumptions \ref{ass:flat} and \ref{ass:CS_init}, the cost function $\bar{Q}^{j}(\cdot)$ satisfies the following properties:
\begin{enumerate}
    \item $\bar{Q}^j(\cdot)$ is non-negative, and equal to $0$ only on $\mathcal{Y}^j_G=\text{conv}\left(\{\mathbf{y}^i_t| \exists i, t: [\bar{\mathbf{y}}^i_t,y^i_{t+R}]\in\mathcal{Y}_G\}\right)$, i.e.,  $\bar{Q}^{j}(\mathbf{y})=0~\forall \mathbf{y}\in\mathcal{Y}^j_G,\ \bar{Q}^{j}(\mathbf{y})\succ 0\ \forall \mathbf{y}\in\bar{\mathcal{CS}}_{\mathbf{y}}^{j}\backslash\mathcal{Y}^j_G.$
    \item $\bar{Q}^{j}(\bar{\mathbf{y}}_{t+1})-\bar{Q}^{j}(\bar{\mathbf{y}}_t)\leq -c(\bar{\mathbf{Y}}_t),\ \forall \bar{\mathbf{y}}_t\in\bar{\mathcal{CS}}_{\mathbf{y}}^{j}$ where $\bar{\mathbf{y}}_{t+1}=\delta(\bar{\mathbf{y}}_t,\bar{y}_{t+R})$ as in \eqref{eq:time_shift} and $\bar{\mathbf{Y}}_t=[\bar{\mathbf{y}}_t,\bar{y}_{t+R}]$.
\end{enumerate}
\end{proposition}
\begin{proof}\textit{Proof deferred to appendix}
\end{proof}
The above proposition shows that $\bar{Q}^{j}(\cdot)$ is in fact a CLF for the dynamics $\bar{\mathbf{y}}_{t+1}=\delta(\bar{\mathbf{y}}_t, \bar{y}_{t+R})$ with input $\bar{y}_{t+R}$ on the set $\mathcal{CS}_{\mathbf{y}}^{j}$. The terminal cost for our Robust MPC is defined as
\begin{align}\label{eq:term_cost}
    P^j(\bar{x}):=\{C\in\mathbb{R}| \exists\bar{\mathbf{y}}\in \bar{\mathcal{CS}}_{\mathbf{y}}^{j-1}, \mathcal{F}_x(\bar{\mathbf{y}})=\bar{x}, \bar{Q}^{j-1}(\bar{\mathbf{y}})=C\}.
\end{align}
\subsection{Robust MPC Feedback Policy and Optimization Problem}\label{LMPC_Pol}
In this section, we consolidate our designed elements (D1)-(D3) and present our Robust MPC policy. At iteration $j$, we use trajectory data $\{\{(x^i_t,u^i_t)\}_{t\geq 0}\}_{i=0}^{j-1}$ to construct outer-approximation $\hat{\mathcal{D}}^j$ of the disturbance support $\mathcal{D}$ by solving \eqref{eq:LMI0} offline, which is used for computing RPI $\mathcal{E}^j$. Then the tightened state and input constraints $\bar{\mathcal{X}}^j, \bar{\mathcal{U}}^j$ are constructed by solving \eqref{eq:tight_const_NLP} offline. The stage cost $c(\cdot)$ is chosen as in \eqref{stage_cost}, the terminal constraint and terminal cost are given by \eqref{eq:term_set}, \eqref{eq:term_cost} respectively. Like the forward-shift operator \eqref{eq:time_shift}, we define the backward-time shift operator,
 \begin{align}\label{eq:b_time_shift}
    \bar{\mathbf{Y}}_{t}&=[\bar{y}_{t},\dots,\bar{y}_{t+R}]\nonumber\\&=\delta^{-}( [\bar{y}_{t+1}, \bar{y}_{t+1},\dots,\bar{y}_{t+R+1}], \bar{y}_{t})\nonumber\\
      &=\delta^{-}(\bar{\mathbf{Y}}_{t+1},\bar{y}_{t})
 \end{align}\normalsize
 Employing this definition, the optimization problem for Robust Output-Lifted LMPC is given by the following:
\begin{equation}\label{eq:OP_RLMPC}
	\begin{aligned}
	J^j_{t}(x^j_t)= \min\limits_{\substack{\bar{\boldsymbol{u}}^j_t, \bar{\boldsymbol{x}}^j_t,\\\bar{\mathbf{y}}^j_{t+N|t}}} \quad & \bar{Q}^{j-1}(\bar{\mathbf{y}}^j_{t+N|t})+\sum\limits_{k=t}^{t+N-1} c(\bar{\mathbf{Y}}^j_{k|t}) \\[1ex]
		\text{s.t.}  \quad & \bar{x}^j_{k+1|t}=f(\bar{x}^j_{k|t},\bar{u}^j_{k|t}), \\
  &\bar{x}^j_{k|t}\in\bar{\mathcal{X}}^j, \bar{u}^j_{k|t}\in\bar{\mathcal{U}}^j~\forall k\in \mathbb{I}_{t}^{t+N-1},\\
    & \bar{\mathbf{Y}}^j_{k|t} =\delta^{-}(\bar{\mathbf{Y}}^j_{k+1|t},h(\bar{x}^j_{k|t}))~\forall k\in\mathbb{I}_t^{t+N-2},\\
    &\bar{x}^j_{t+N|t}=\mathcal{F}_x(\bar{\mathbf{y}}^j_{t+N|t}),~ \bar{\mathbf{y}}^j_{t+N|t}\in\bar{\mathcal{CS}}_{\mathbf{y}}^{j-1},\\
        & \bar{\mathbf{Y}}^j_{t+N-1|t}=[h(\bar{x}^j_{t+N-1|t}),\bar{\mathbf{y}}^j_{t+N|t}],\\
        & x^j_t-\bar{x}^j_{t|t} \in \mathcal{E}^j,\bar{x}^j_{t|t}=\bar{x}^{j\star}_{t|t-1}
	\end{aligned}
\end{equation}\normalsize
where $\boldsymbol{u}^j_t = [\bar{u}^j_{t|t},\ldots,\bar{u}^j_{t+N-1|t}],\  \boldsymbol{x}^j_t = [\bar{x}^j_{t|t},\ldots,\bar{x}^j_{t+N|t}]$ and $x^j_t$ is the state of the system at time $t$.  The control policy is obtained by solving \eqref{eq:OP_RLMPC} online and using the optimal solution $\bar{u}^j_t=\bar{u}^{j\star}_{t|t}, \bar{x}^j_t=\bar{x}^{j\star}_{t|t}$, as
\begin{align}\label{eq:LMPC}
u^j_t=\pi^j(x^j_t)=\bar{u}^{j}_{t}+K(x^j_t-\bar{x}^{j}_{t})
\end{align}
The optimal nominal state-input trajectory $\{(\bar{x}^{j\star}_{t|t}, \bar{u}^{j\star}_{t|t})\}_{t\geq 0}$ satisfies \eqref{eq:nom_spec}, by feasibility of \eqref{eq:OP_RLMPC}. Thus, this trajectory can be used for constructing the set $\bar{\mathcal{CS}}_{\mathbf{y}}^{j}$ and function $\bar{Q}^{j}(\cdot)$ for iteration $j+1$. We summarize the iterative policy synthesis in Algorithm~\ref{algo:ROLMPC}. Next, we will analyze the properties of the system \eqref{eq:sysdyn} in closed-loop with \eqref{eq:LMPC}.
\begin{remark} \textit{The nominal state $\bar{x}^j_t=x^{j\star}_{t|t}$ for \eqref{eq:OP_RLMPC} at time $t$ is obtained from the solution of \eqref{eq:OP_RLMPC} at time $t-1$. Consequently, the solution to \eqref{eq:OP_RLMPC} is the same for any $x^j_t\in\mathcal{E}^j\oplus x^{j\star}_{t|t}$. To incorporate feedback from $x^j_t$, the constraints $x^j_t-\bar{x}^j_{t|t}\in\mathcal{E}^j$, $\bar{x}^j_{t|t}=\bar{x}^{j\star}_{t|t-1}$ can be replaced with $x^j_t-\bar{x}^j_{t|t}\in\mathcal{E}^j$, $\bar{x}^j_{t|t}=f(\bar{x}^{j\star}_{t-1|t-1}, \bar{u}^j_{t-1|t}), \bar{u}^j_{t-1|t}\in\bar{\mathcal{U}}^j$, where the nominal state $\bar{x}^j_t=x^{j\star}_{t|t}$ and nominal input $\bar{u}^j_{t-1}=\bar{u}^{j\star}_{t-1|t}$ are re-computed to generate a nominal state-input trajectory satisfying \eqref{eq:nom_spec}.}
\end{remark}
\begin{algorithm}[!h]
    \caption{Iterative Robust Output-lifted LMPC}\label{algo:ROLMPC}
    \DontPrintSemicolon
    \SetKwInOut{KwIn}{Input}
    \SetKwInOut{KwOut}{Output}
    \SetKwFunction{Nom}{nominal\_states}
    \SetKwFunction{Supp}{d\_bound}
    \SetKwFunction{Einv}{err\_inv}
    \SetKwFunction{tc}{tight\_cons}
    \SetKwFunction{CSy}{construct\_terminal\_set}
    \SetKwFunction{ROLMPC}{Robust\_Output-lifted\_LMPC}
    \KwIn{$\{(x^0_t,u^0_t)\}_{t\geq 0}$\text{ satisfying }\eqref{eq:iter_spec}\text{, Assumption }\ref{ass:E_init},\\$\{(\bar{x}^0_t,\bar{u}^0_t)\}_{t\geq 0}$\text{ satisfying }\eqref{eq:nom_spec}\text{, Assumption }\ref{ass:CS_init},\\\text{ Error policy }$\kappa(e_t)=Ke_t$, $\text{max}\_\text{iters}$ }
    \KwOut{$\{\{(x^j_t,u^j_t)\}_{t\geq 0}\}_{j=1}^{\text{max}\_\text{iters}}$\text{ satisfying }\eqref{eq:iter_spec}}
    \SetKwProg{Fn}{Procedure}{:}{\KwRet $\{\{(x^j_t,u^j_t)\}_{t\geq 0}\}_{j=1}^{\text{max}\_\text{iters}}$}
    \Fn{\ROLMPC}{
    $\bar{\mathcal{CS}}^0_{\mathbf{y}}\leftarrow\emptyset$\\
    \For{$j=1$ \KwTo $\text{max}\_\text{iters}$}{
    \textbf{\textit{(Offline)}}\\
    \textbf{Step 1}:~Construct $\hat{\mathcal{D}}^j$ by solving \eqref{eq:LMI0} and compute $\mathcal{E}^j$ by Assumption~\ref{ass:einv}\\
    \textbf{Step 2}:~Construct $\bar{\mathcal{X}}^j,\bar{\mathcal{U}}^j$ by solving \eqref{eq:tight_const_NLP} \\
    \textbf{Step 3}:~Construct $\bar{\mathcal{CS}}^{j}_{\mathbf{y}}$ from \eqref{eq:CSy_def} and
        $\bar{Q}^{j}(\cdot)$ from \eqref{eq:conv_cf}\\
        \textbf{\textit{(Online)}}\\
         ~$x^j_0\leftarrow x_S, t\leftarrow 0$\\
        \While{$x^j_t\not\in\mathcal{X}_G$}{
        Solve \eqref{eq:OP_RLMPC} and store $\bar{x}^j_t\leftarrow \bar{x}^{j\star}_{t|t},\bar{u}^j_{t-1}\leftarrow \bar{u}^{j\star}_{t-1|t}$\\
        Apply \eqref{eq:LMPC} and store $u^j_{t}\leftarrow \pi^{j}_{t}(x^j_t)$\\
        Measure $x^j_{t+1}$\\ 
        $t\leftarrow t+1$
        }}}\end{algorithm}
\section{Properties of Proposed Strategy}\label{sec:LMPC_analysis}
In this section, we establish the closed-loop properties of the system trajectories with the proposed Robust Output-lifted LMPC, and examine the system performance across iterations.

Theorem \ref{thm:rf} establishes the recursive feasibility of optimization problem \eqref{eq:OP_RLMPC} for system \eqref{eq:sysdyn} in closed-loop with the LMPC policy \eqref{eq:LMPC}. We show this by leveraging the recursive definition of $\bar{\mathcal{CS}}_{\mathbf{y}}^{j}$ and the result of Proposition~\ref{prop:CS_CI}. 
\begin{theorem}\label{thm:rf}[\textbf{Recursive Feasibility}] Given Assumptions \ref{ass:dQC}-\ref{ass:CS_init}, the optimization problem \eqref{eq:OP_RLMPC} is feasible for the system \eqref{eq:sysdyn} in closed-loop with the policy \eqref{eq:LMPC} $\forall t\geq 0$ and for all iterations $j\geq 1$. Consequently, $x^j_t\in\mathcal{X}$, $u^j_t\in\mathcal{U}$ $\forall t\geq 0$, $\forall j\geq 1$. 
\end{theorem}
\begin{proof}For any iteration $j\geq 1$, suppose that the problem \eqref{eq:OP_RLMPC} is feasible at time $t\geq 1$. Let the state-input trajectory corresponding to the optimal solution  of \eqref{eq:OP_RLMPC} be \small\begin{equation}\label{eq:opt_sol_t}
    \{\bar{x}^{j\star}_{t|t},\bar{u}^{j\star}_{t|t}, \bar{x}^{j\star}_{t+1|t},\bar{u}^{j\star}_{t+1|t},\dots, \bar{x}^{j\star}_{t+N|t}, \bar{\mathbf{y}}^{j\star}_{t+N|t}\}.
\end{equation}\normalsize  Applying the control $u^j_t=\bar{u}^{j\star}_{t|t}+K(x^j_t-\bar{x}^{j\star}_{t|t})$ to system \eqref{eq:sysdyn} yields $x^j_{t+1}$ such that $x^j_{t+1}-\bar{x}^{j\star}_{t+1|t}\in\mathcal{E}^j,~\forall d_t\in D(x^j_t,u^j_t)$ because $x^j_t-\bar{x}^{j\star}_{t|t}\in\mathcal{E}^j$ ($\because$ Proposition \ref{prop:errinv}), where $\bar{x}^{j\star}_{t+1|t}=f(\bar{x}^j_{t}, \bar{u}^{j\star}_{t|t})$ and $\bar{x}^j_t=\bar{x}^{j\star}_{t|t}$. We also have
\small$$\bar{x}^{j\star}_{t+N|t}=\mathcal{F}_x(\bar{\mathbf{y}}^{j\star}_{t+N|t}),\ \bar{\mathbf{y}}^{j\star}_{t+N|t}\in\bar{\mathcal{CS}}_{\mathbf{y}}^{j}.$$\normalsize
From Proposition \ref{prop:CS_CI}, we have $\bar{x}^{j\star}_{t+N|t}\in\bar{\mathcal{X}}^j$, and $\exists\bar{y}'$ such that $\bar{\mathbf{y}}'=\delta(\bar{\mathbf{y}}_{t+N|t}^{j\star},\bar{y}')\in\bar{\mathcal{CS}}_{\mathbf{y}}^j$, $\bar{u}'=\mathcal{F}_u([\bar{\mathbf{y}}^{j\star}_{t+N|t},\bar{y}'])\in\bar{\mathcal{U}}^j$ and $\bar{x}'=\mathcal{F}_x(\bar{\mathbf{y}}')=f(\bar{x}^{j\star}_{t+N|t},\tilde{u})\in\bar{\mathcal{X}}^j$. Now consider the following state-input trajectory\small
\begin{equation}\label{eq:shift_sol_t}
    \{x^{j\star}_{t+1|t},\bar{u}^{j\star}_{t+1|t},\dots, \bar{x}^{j\star}_{t+N|t},\bar{u}',\bar{x}', \bar{\mathbf{y}}'\}
\end{equation}\normalsize
and see that this is feasible for problem \eqref{eq:OP_RLMPC} at time $t+1$. \\
We have shown that feasibility of the LMPC problem~\eqref{eq:OP_RLMPC} at time $t\geq 1$ implies feasibility of the LMPC problem~\eqref{eq:OP_RLMPC} at time $t+1$. For $t=0$ and any $j\geq1$, we have $x^j_t=x^{j-1}_t=x_S$ and $\bar{\mathcal{X}}^j\supset\bar{\mathcal{X}}^{j-1}$, $\bar{\mathcal{U}}^j\supset\bar{\mathcal{U}}^{j-1}$, $\bar{\mathcal{CS}}^j_{\mathbf{y}}\supset\bar{\mathcal{CS}}^{j-1}_{\mathbf{y}}$. Thus, the solution to \eqref{eq:OP_RLMPC} from iteration $j-1$ at $t=0$ is feasible for iteration $j$, with the initialization for $j=1$ given by Assumption \ref{ass:CS_init}. Induction on time $t$ proves the persistent feasibility of \eqref{eq:OP_RLMPC} $\forall t\geq 0$, $\forall j\geq 1$.

Thus, $x^j_t-\bar{x}^j_t\in\mathcal{E}^j$, $\bar{x}^j_t\in\bar{\mathcal{X}}^j\subset\mathcal{X}\ominus\mathcal{E}^j$, $\bar{u}^j_t\in\bar{\mathcal{U}}^j\subset\mathcal{U}\ominus K\mathcal{E}^j$ $\forall t\geq 0, \forall j\geq 1$. Since $(\mathcal{X}\ominus\mathcal{E}^j)\oplus\mathcal{E}^j\subset \mathcal{X}, (\mathcal{U}\ominus K\mathcal{E}^j)\oplus K\mathcal{E}^j\subset \mathcal{U}$, we have $x^j_t\in\mathcal{X}, u^j_t\in\mathcal{U}$.
\end{proof}

We also establish convergence of the closed-loop state trajectories of \eqref{eq:sysdyn} to the set $\mathcal{X}_G$. First,  we show that if $\lim_{t\rightarrow\infty}\text{dist}_{\mathcal{Y}_G}(\bar{\mathbf{Y}}^j_t)=0$ then $\lim_{t\rightarrow\infty}\text{dist}_{\mathcal{X}_G\ominus\mathcal{E}^1}(\bar{x}^j_t)=0$. We use this result to finally show that $\lim_{t\rightarrow\infty}\text{dist}_{\mathcal{X}_G}(x^j_t)=0$ in the proof of Theorem~\ref{thm:convergence}.
\begin{lemma}\label{lem:yconv_xconv}
Given Assumption \ref{ass:flat}, if the trajectory of lifted outputs $\{\bar{\mathbf{Y}}_t\}_{t\geq0}$ for system \eqref{eq:nom_dyn} converges to the set $\mathcal{Y}_G$, then the nominal state trajectory $\{\bar{x}_t\}_{t\geq 0}$ converges to $\mathcal{X}_G\ominus\mathcal{E}^1$, i.e.,
$$\lim_{t\rightarrow\infty}\text{dist}_{\mathcal{Y}_G}(\bar{\mathbf{Y}}^i_t)=0\Rightarrow \lim_{t\rightarrow\infty}\text{dist}_{\mathcal{X}_G\ominus\mathcal{E}^1}(\bar{x}_t)=0$$
\end{lemma}
\begin{proof}
Define the set $\mathcal{Y}^x_G=\{\mathbf{y}| [\mathbf{y},y]\in\mathcal{Y}_G\}$, which is just a linear projection of $\mathcal{Y}_G$ to obtain the first $R$ outputs, and see that $\lim_{t\rightarrow\infty}\text{dist}_{\mathcal{Y}_G}(\bar{\mathbf{Y}}_t)=0 \Rightarrow \lim_{t\rightarrow\infty}\text{dist}_{\mathcal{Y}^x_G}(\bar{\mathbf{y}}_t)=0$ because of the continuity of the projection. Then from the fact that the image of the flat map \eqref{eq:flatclass_x} is unique (Definition~\ref{def:diffFlat}) and the continuity of $\mathcal{F}_x(\cdot)$ (Assumption \ref{ass:flat}(A)), we have
\begin{align*}
    \lim_{t\rightarrow\infty}\text{dist}_{\mathcal{X}_G\ominus\mathcal{E}^1}(\bar{x}_t)&=\lim_{t\rightarrow\infty}\text{dist}_{\mathcal{X}_G\ominus\mathcal{E}^1}(\mathcal{F}_x(\bar{\mathbf{y}}_t))\\
    &=\lim_{t\rightarrow\infty}\text{dist}_{\mathcal{F}_x^{-1}(\mathcal{X}_G\ominus\mathcal{E}^1)}(\bar{\mathbf{y}}_t).
\end{align*}
From the definition of $\mathcal{Y}_G, \mathcal{Y}^x_G$, we know that $\bar{\mathbf{y}}\in\mathcal{Y}^x_G\Rightarrow \mathcal{F}_x(\bar{\mathbf{y}})=\bar{x}\in\mathcal{X}_G\ominus\text{Box}(\mathcal{E}^1)\subset\mathcal{X}\ominus\mathcal{E}^1$. Thus, $\mathcal{F}_x(\mathcal{Y}^x_G)\subseteq\mathcal{X}_G\ominus\mathcal{E}^1\Rightarrow\mathcal{Y}^x_G\subseteq\mathcal{F}^{-1}_x(\mathcal{X}_G\ominus\mathcal{E}^1)$ and so $\lim_{t\rightarrow\infty}\text{dist}_{\mathcal{Y}^x_G}(\bar{\mathbf{y}}_t)=0\Rightarrow \lim_{t\rightarrow\infty}\text{dist}_{\mathcal{F}^{-1}_x(\mathcal{X}_G\ominus\mathcal{E}^1)}(\bar{\mathbf{y}}_t)=0$.
$$ \therefore \lim_{t\rightarrow\infty}\text{dist}_{\mathcal{X}_G\ominus\mathcal{E}^1}(\bar{x}_t)=0$$
\end{proof}
\begin{theorem}\label{thm:convergence}[\textbf{Convergence}] Given Assumptions \ref{ass:dQC}-\ref{ass:CS_init}, for any iteration $j\geq 1$, the state trajectory of system \eqref{eq:sysdyn} in closed-loop with control $\eqref{eq:LMPC}$ converges to the set $\mathcal{X}_G$,
$\lim_{t\rightarrow\infty}\text{dist}_{\mathcal{X}_G}(x^j_t)=0$.
\end{theorem}
\begin{proof}
We adopt the same notation as the proof of theorem \ref{thm:rf}. Using the feasibility of \eqref{eq:shift_sol_t} for the problem \eqref{eq:OP_RLMPC} at time $t+1$ and the fact that $J^{j}_{t+1}(x^j_{t+1})$ is the optimal cost at time $t+1$, we get\small
\begin{align}\label{eq:cost_decrease}
    J^j_{t+1}(x^j_{t+1})&\leq \bar{Q}^{j}(\mathbf{y}')+c(\bar{\mathbf{Y}}_{N|t}^\star)+\sum\limits_{k=1}^{N-1}c(\bar{\mathbf{Y}}^\star_{k|t})\nonumber\\
    &\leq \bar{Q}^{j}(\bar{\mathbf{y}}^\star_{N|t})+\sum\limits_{k=1}^{N-1}c(\bar{\mathbf{Y}}^\star_{k|t})\nonumber\\
    &=J^j_{t}(x^j_t)-c(\bar{\mathbf{Y}}^\star_{t|t}).
\end{align}\normalsize
The feasibility of the problem \eqref{eq:OP_RLMPC} (guaranteed by Theorem~\ref{thm:rf}) and positive definiteness of $c(\cdot)$ imply that the sequence $\{J^j_{t}(x^j_t)\}_{t\geq 0}$ is non-increasing. Moreover, positive definiteness of $\bar{Q}^{j}(\cdot)$ (by Proposition~\ref{prop:Q_CLF}) further implies that the sequence is lower bounded by $0$. Thus the sequence converges and taking limits on both sides of \eqref{eq:cost_decrease} gives\small
$$0\leq \lim_{t\rightarrow\infty}-c(\bar{\mathbf{Y}}^j_t)\leq 0\Rightarrow \lim_{t\rightarrow\infty}c(\bar{\mathbf{Y}}^j_t)=0$$\normalsize
By continuity of $c(\cdot)$, we have that $\lim_{t\rightarrow\infty}c(\bar{\mathbf{Y}}^j_t)=0\Leftrightarrow \lim_{t\rightarrow\infty}\text{dist}_{c^{-1}(0)}(\bar{\mathbf{Y}}^j_t)=0$. From \eqref{stage_cost}, we know $c^{-1}(0)=\mathcal{Y}_G$ and thus, $\lim_{t\rightarrow\infty}\text{dist}_{\mathcal{Y}_G}(\bar{\mathbf{Y}}^j_t)=0$ and $\lim_{t\rightarrow\infty}\text{dist}_{\mathcal{X}_G\ominus\mathcal{E}^1}(\bar{x}^j_t)=0$ by lemma \ref{lem:yconv_xconv}. 
 By Proposition \ref{prop:errinv},  we also have $x^j_t-\bar{x}^j_t\in\mathcal{E}^j\subseteq\mathcal{E}^1$. Using the facts $(\mathcal{X}_G\ominus\mathcal{E}^1)\oplus\mathcal{E}^1\subseteq\mathcal{X}_G$ and $\text{dist}_S(\cdot)\geq 0$, we have
\small
\begin{align*}
    \lim_{t\rightarrow\infty}\text{dist}_{\mathcal{X}_G}(x^j_t)=&\lim_{t\rightarrow\infty}\text{dist}_{\mathcal{X}_G}(\bar{x}^j_t+x^j_t-\bar{x}^j_t)\\
    \leq&\lim_{t\rightarrow\infty}\text{dist}_{(\mathcal{X}_G\ominus\mathcal{E}^1)\oplus\mathcal{E}^1}(\bar{x}^j_t+x^j_t-\bar{x}^j_t)\\
    =&0\\
    \Rightarrow \lim_{t\rightarrow\infty}\text{dist}_{\mathcal{X}_G}(x^j_t)=&0
\end{align*}
\normalsize
\end{proof}

We conclude our theoretical analysis of the proposed Robust MPC \eqref{eq:LMPC} with the following theorem. We state and prove that the closed-loop costs of system trajectories in closed-loop with the LMPC do not increase with iterations if the system starts from the same state, i.e., $x^j_0=x_S~\forall j\geq 0$.
\begin{theorem}\label{thm:cost_imp}[\textbf{Performance Improvement}]
Given Assumptions \ref{ass:dQC}-\ref{ass:CS_init}, the cost of the trajectories of system \eqref{eq:sysdyn} in closed-loop with control \eqref{eq:LMPC} does not increase with iterations,
$$j_2>j_1\Rightarrow J^{j_2}_{0\rightarrow \infty}(x_S)\leq J^{j_1}_{0\rightarrow \infty}(x_S)$$
where $J^{j}_{0\rightarrow \infty}(x_S)=\mathcal{C}^j_0$.
\end{theorem}
\begin{proof}
 The cost of the trajectory in iteration $j-1$ is given by\small
\begin{align*}
    J^{j-1}_{0\rightarrow \infty}(x_S)&=\sum\limits_{t\geq 0} c(\bar{\mathbf{Y}}^{j-1}_t)\\
    &=\sum\limits_{t=0}^{N-1}c(\bar{\mathbf{Y}}^{j-1}_t)+\mathcal{C}^{j-1}_N\\
    &\geq\sum\limits_{t=0}^{N-1}c(\bar{\mathbf{Y}}^{j-1}_t)+\bar{Q}^{j}(\bar{\mathbf{y}}^{j-1}_N)\\
    &\geq J^j_{0}(x_S)
\end{align*}\normalsize
The second to last inequality comes from the definition of $\bar{Q}^{j}(\cdot)$ in \eqref{eq:conv_cf} while the last inequality comes from optimality of problem \eqref{eq:OP_RLMPC} in the $j$th iteration starting from $x^j_0=x_S$.\\
Now we use inequality \eqref{eq:cost_decrease} repeatedly to derive\small
\begin{align*}
    J^j_{0}(x_S)\geq& c(\bar{\mathbf{Y}}^j_0)+J^j_{1}(x^j_1)\\
    \geq& c(\bar{\mathbf{Y}}^j_0)+c(\bar{\mathbf{Y}}^j_1)+J^{j}_{2}(x^j_2)\\
    \geq& \lim_{t\rightarrow\infty}( \sum\limits_{k=0}^{t-1}c(\bar{\mathbf{Y}}^j_t)+J^j_{t}(x^j_t))\\
    \geq &\lim_{t\rightarrow\infty} \sum\limits_{k=0}^{t-1}c(\bar{\mathbf{Y}}^j_t)\\
    =&J^j_{0\rightarrow\infty}(x_S)
\end{align*}\normalsize
Thus,
\small$$J^{j-1}_{0\rightarrow\infty}(x_S)\geq J^j_{0\rightarrow N}(x_S)\geq J^{j}_{0\rightarrow\infty}(x_S)$$\normalsize
The desired statement easily follows from above.   
\end{proof}

\section{Numerical Example: Kinematic Bicycle in Frenet Frame}\label{sec:ex}
In this section, we demonstrate our approach for constrained optimal control of a kinematic bicycle in the Frenet frame. The code for this example is hosted at  \url{https://github.com/shn66/ROLMPC}.
\subsection{Problem Formulation}\label{ssec:kb_prob_form}
We solve a constrained optimal control problem for driving a kinematic bicycle over a chicane into a goal set $\mathcal{X}_G$. The dynamics $f(\cdot)$ of the bicycle are described in the Frenet frame, \begin{align}
    s_{t+1}&=s_{t}+\dfrac{v_t \cos(e_{\psi t})}{1-e_{y,t}C(s_t)}\text{dt}+d_{1,t}\nonumber\\
    e_{y,t+1}&=e_{y,t}+v_t\sin(e_{\psi, t})\text{dt}+d_{2,t}\nonumber\\
e_{\psi,t+1}&=e_{\psi,t}+v_t\left(\frac{\tan(\delta_t)}{L_{RF}}-\dfrac{\cos(e_{\psi t})C(s_t)}{1-e_{y,t}C(s_t)}\right)\text{dt}+d_{3,t}\nonumber
\end{align}
where time-step $\text{dt}=0.2 s$, and state $x_t=[s_t, e_{y,t}, e_{\psi, t}]$ consists of the longitudinal abcissa $s_t$, the lateral offset $e_{y,t}$ and the heading alignment error $e_{\psi, t}$ w.r.t the road center-line. The disturbance $d_t=[d_{1,t},d_{2,t},d_{3,t}]$ captures bounded process noise, errors from discretization and conversion between the Euclidean and Frenet frame. The inputs $u_t=[v_t,\delta_t]$ are the speed $v_t$ of the rear axle and steering angle $\delta_t$ of the front axle.  $L_{RF}=4m$ is the wheelbase of the vehicle, and it is assumed that the centre of gravity is on the rear axle. The center-line is given by a chicane with curvature $C(s_t)=\frac{1}{10\pi}\tan^{-1}(100-\frac{1}{2}s^2_t)$. The constraints are given by the box sets as in Assumption~\ref{ass:box}:
\begin{align}
    \mathcal{X}&=\left\{[s,e_y,e_{\psi}]\middle\vert s\in[-2,60],e_y\in[-4.5,4.5],e_\psi\in[-\frac{\pi}{3},\frac{\pi}{3}]\right\}\nonumber\\
    \mathcal{U}&=\left\{[v,\delta]\middle\vert\ v\in[0,18], \delta\in[-\frac{\pi}{2},\frac{\pi}{2}]\right\},\nonumber
\end{align}
and the goal set is $\mathcal{X}_G=\{[s,e_y,e_\psi]\in\mathcal{X}| s\geq 40\}$.

We use our Robust Output-lifted LMPC to iteratively approximate the solution of the following optimal control problem
\begin{equation}\label{eq:inf_OP_frenet_bike}
	\begin{aligned}
\min\limits_{\{\pi_t(\cdot)\}_{t\geq 0}} \quad & \displaystyle \sum\limits_{t=0}^{\infty} \max\{40-s_t,0\} \\[1ex]
		\text{s.t.}\quad  & x_{t+1}=f( x_{t},\pi_t(x_{t}))+d_t, \\
    & x_{t}\in\mathcal{X},~\pi_t(x_t)\in\mathcal{U},\\
    & x_0=x_S
	\end{aligned}
\end{equation}
for the kinematic bicycle starting from $x_S=[0,1,0]$.
To apply Algorithm \ref{algo:ROLMPC}, we verify that Assumptions~\ref{ass:dQC}, \ref{ass:flat} and \ref{ass:einv} are satisfied. First, we describe the lifted output $\bar{\mathbf{Y}}$ and associated maps $\mathcal{F}_x(\cdot), \mathcal{F}_u(\cdot)$ for the kinematic bicycle, and obtain the bounding functions $\mathcal{F}^{\cup}(\cdot), \mathcal{F}^{\cap}(\cdot)$ for verifying Assumption~\ref{ass:flat}.  Second, we obtain an uncertainty description (as in Assumption \ref{ass:dQC}) for the additive disturbance $d_t$ using data. To verify Assumption~\ref{ass:einv}, we describe a procedure for constructing the error invariant $\mathcal{E}$ for the error dynamics \eqref{eq:error_dyn}, and choosing a fixed linear policy $\kappa(e_t)=Ke_t$. The three steps are detailed below.

\subsubsection*{Lifted Outputs} The difference flat output for the nominal kinematic bicycle model are given by $\bar{s}_t, \bar{e}_{y,t}$. The lifted output and associated maps for the nominal system are given by 

{\begingroup
\allowdisplaybreaks
\small
\begin{align*}
    \bar{y}_t=[\bar{y}_{1,t}\ \bar{y}_{2,t}]^\top&=[\bar{s}_t\ \bar{e}_{y,t}]^\top,\ \bar{\mathbf{Y}}_t=[\bar{y}_t,\bar{y}_{t+1},\bar{y}_{t+2}]\\
\mathcal{F}_x(\bar{y}_t,\bar{y}_{t+1})&=\begin{bmatrix}\bar{y}_t\\\tan^{-1}\left(\dfrac{\bar{y}_{2,t+1}-\bar{y}_{2,t}}{(1-\bar{y}_{2,t}C(\bar{y}_{1,t}))(\bar{y}_{1,t+1}-\bar{y}_{1,t})}\right)\end{bmatrix}\nonumber\\
    \mathcal{F}_u(\bar{y}_t,\bar{y}_{t+1},\bar{y}_{t+2})&=[\bar{v}_t\ \bar{\delta}_t]^\top
    \end{align*}
\begin{align*}
    \bar{v}_t&=\frac{1}{dt}\left\Vert\begin{bmatrix}(1-\bar{y}_{2,t}C(\bar{y}_{1,t}))(\bar{y}_{1,t+1}-\bar{y}_{1,t})\\ \bar{y}_{2,t+1}-\bar{y}_{2,t}\end{bmatrix}\right\Vert_2\\
    \bar{\delta}_t&=\tan^{-1}\Big([0\ 0\ \frac{L_{RF}}{dt\bar{v}_t}]\big(\mathcal{F}_x(\bar{y}_{t+1},\bar{y}_{t+2})-\mathcal{F}_x(\bar{y}_{t},\bar{y}_{t+1})\big)\nonumber\\
    &\quad (\bar{y}_{1,t+1}-\bar{y}_{1,t})C(\bar{y}_{1,t})\Big)
\end{align*}
\normalsize
\endgroup}
Next, we propose bounding functions $\mathcal{F}^{\cup}(\cdot)$, $\mathcal{F}^{\cap}(\cdot)$ such that they are quasiconvex and quasiconcave respectively, with $\mathcal{F}^\cap(\bar{\mathbf{Y}})\leq \mathcal{F}(\bar{\mathbf{Y}})\leq \mathcal{F}^\cup(\bar{\mathbf{Y}})$, as required by Assumption \ref{ass:flat}\eqref{ass:flatmap_convex}. For the nominal positions $\bar{s}_t,\bar{e}_{y,t}$, the bounding functions are trivially given by $\mathcal{F}^1_x(\cdot), \mathcal{F}_x^2(\cdot)$ because linear functions are both quasiconvex and quasiconcave. 
{\small{
\begin{align*}
    &[\mathcal{F}^{1,\cap}(\bar{\mathbf{Y}}_t),\mathcal{F}^{2,\cap}(\bar{\mathbf{Y}}_t)]=[\mathcal{F}^{1,\cup}(\bar{\mathbf{Y}}_t),\mathcal{F}^{2,\cup}(\bar{\mathbf{Y}}_t)]=\bar{y}_t
    \end{align*}}}
    For the bounding functions corresponding to $\bar{e}_{\psi,t}$, we use the system constraints to bound $(1-\bar{y}_{2,t}C(\bar{y}_{1,t}))\in [\frac{31}{40},\frac{49}{40}]$, to construct bounding functions as:
{\small{    
    \begin{align*}
    &\mathcal{F}^{3,\cap}(\bar{\mathbf{Y}}_t)=\tan^{-1}\left(\min\{\frac{(\bar{y}_{2,t+1}-\bar{y}_{2,t})}{\frac{49}{40}(\bar{y}_{1,t+1}-\bar{y}_{1,t})}, \frac{(\bar{y}_{2,t+1}-\bar{y}_{2,t})}{\frac{31}{40}(\bar{y}_{1,t+1}-\bar{y}_{1,t})}\}\right) \\
    &\mathcal{F}^{3,\cup}(\bar{\mathbf{Y}}_t)=\tan^{-1}\left(\max\{\frac{(\bar{y}_{2,t+1}-\bar{y}_{2,t})}{\frac{49}{40}(\bar{y}_{1,t+1}-\bar{y}_{1,t})}, \frac{(\bar{y}_{2,t+1}-\bar{y}_{2,t})}{\frac{31}{40}(\bar{y}_{1,t+1}-\bar{y}_{1,t})}\}\right) 
\end{align*}}}
where $\mathcal{F}^{3,\cup}(\cdot), \mathcal{F}^{3,\cap}(\cdot)$ can be verified to be quasiconvex and quasiconcave respectively by using composition rules of quasilinear functions. Similarly for the inputs, we get: 
{\small{
\begin{align*}
    &\mathcal{F}^{4,\cap}(\bar{\mathbf{Y}}_t)=0, 
    \mathcal{F}^{4,\cup}(\bar{\mathbf{Y}}_t)=\frac{1}{dt}\left\Vert\begin{bmatrix}\frac{49}{40}(\bar{y}_{1,t+1}-\bar{y}_{1,t})\\ \bar{y}_{2,t+1}-\bar{y}_{2,t}\end{bmatrix}\right\Vert_2\\
    &\mathcal{F}^{5,\cap}(\bar{\mathbf{Y}}_t)=-\frac{\pi}{2}, 
    \mathcal{F}^{5,\cup}(\bar{\mathbf{Y}}_t)=\frac{\pi}{2}
\end{align*}
}}

\subsubsection*{Uncertainty Modeling} The disturbance $d_t$ is assumed to lie in a state and input dependent set $D(x,u)$, modelled as $D(x,u)=\{d| \exists w: d=d(x,u)+w, ||w||_{2}\leq\gamma\}$, where the function $d(x,u)$ is unknown, but assumed to be $L-$Lipschitz. The conditions on $D(x,u)$ for Assumption \ref{ass:dQC} are satisfied as shown in Example \ref{eg:lip_bnded}.
The constants $L$, $\gamma$ are estimated as follows:
\begin{enumerate}
    \item Sample system transitions to obtain  data-set $\mathcal{T}=\{(x_i,u_i,d_i)\}_{i=1}^{T}$ 
    \item The true Lipschitz constant $L$ and bound $\gamma$ are one of the minimizers of the following semi-infinite, multi-objective optimization problem: {\small{\begin{align}\min
    \left\{(\tilde{L},\tilde{\gamma})\in\mathbb{R}^2_+ \middle\vert\begin{aligned}
    &\Vert \bar{d}-\bar{d}'\Vert_2\leq \left|\tilde{L}\left\Vert \begin{bmatrix}x-x'\\u-u'\end{bmatrix}\right\Vert_2+2\tilde{\gamma}\right|,\\
    &\forall \bar{d}\in D(x,u),\bar{d}'\in D(x',u'),\\ &\forall (x,u),(x',u')\in\mathbb{R}^{n+m}
    \end{aligned}\right\}\nonumber
    \end{align}}}
    \item We solve for $L,\gamma$ for the scalarized objective $\tilde{L}+\tilde{\gamma}$, and approximate the semi-infinite optimization via the scenario approach using the data-set $\mathcal{T}$ to obtain the LP: 
    {\small{\begin{align}
\begin{aligned}\min_{\tilde{L}>0,\tilde{\gamma}>0}&\tilde{L}+\tilde{\gamma}\\
        \text{s.t}&\ \Vert d_i-d_j\Vert_2\leq \tilde{L}\left\Vert \begin{bmatrix}x_i-x_j\\u_i-u_j\end{bmatrix}\right\Vert_2+2\tilde{\gamma},\\
    &~~\forall (x_i,u_i,d_i),(x_j,u_j,d_j)\in\mathcal{T}\end{aligned}\nonumber
    \end{align}}}
The approximated constants obtained after solving the linear program were $\hat{L}=0.0329$, $\hat{\gamma}=0.1640$. Using \cite[Corollary 6]{alamorandom} and linearity of the semi-infinite constraints, it can be shown that sampled constraints provide an \textit{inner approximation} of the actual feasible set for the semi-infinite problem with high-confidence. Thus, the event $E:=\{\hat{L}\geq L, \hat{\gamma}\geq \gamma\}$ holds with high probability, and the statements of Theorems \ref{thm:D_cnstrct}--\ref{thm:cost_imp} hold conditioned on $E$ (which is conventional as noted in Remark \ref{rem:approx_constants}).
\end{enumerate}

\subsubsection*{Error Invariant and Error Policy} To construct the error invariant and the error policy for Assumption \ref{ass:einv}, we linearise the bicycle dynamics about $x_{ref}=(s, e_y,e_{\psi})=(\sqrt{200},0,0)$ $u_{ref}=(v,\delta)=(10, 0)$, and obtain bounds on the higher-order terms using the system constraints to give the linearized dynamics $x_{t+1}=f(x_{ref},u_{ref})+A(x_t-x_{ref})+B(u_t-u_{ref})+n_t+d_t$, where $n_t\in\mathcal{W}^l$ is the linearisation error and $d_t\in\hat{\mathcal{D}}$ corresponds to the error due to unmodelled dynamics. Similarly, the nominal dynamics are given as $\bar{x}_{t+1}=f(x_{ref},u_{ref})+A(\bar{x}_t-x_{ref})+B(\bar{u}_t-u_{ref})+\bar{n}_t$, and so, the error dynamics are $e_{t+1}=Ae_t+B\kappa(e_t)+d_t+[I_n\ -I_n][n^\top_t\ \bar{n}^\top_t]^\top$. For the combined disturbance $\bar{\mathcal{D}}=\hat{\mathcal{D}}\oplus [I_n\ -I_n](\mathcal{W}^l\times \mathcal{W}^l)$, the RPI is computed by fixing $\kappa(e_t)=K_{LQR}e_t$ and setting $\mathcal{E}_{RPI}=\bigoplus_{i=0}^\infty (A+BK_{LQR})^i\bar{\mathcal{D}}$. The cost matrices $Q, R$ for the LQR policy are tuned such that $\mathcal{E}_{RPI}\subset\mathcal{X}$ and $K_{LQR}\mathcal{E}_{RPI}\subset\mathcal{U}$.

\subsection{Results}
We implement the proposed Robust Output-lifted LMPC strategy as described in Algorithm~\ref{algo:ROLMPC}  for optimal control of the kinematic bicycle from Section~\ref{ssec:kb_prob_form}.  The results emphasize the following aspects of our approach:

\subsubsection{Iterative Learning}
The trajectory data across iterations $i=0,..,j-1$ is used for constructing the disturbance support $\hat{\mathcal{D}}^j$, the terminal set $\bar{\mathcal{X}}^j_N$ and trajectory cost estimate via the terminal cost $P^j(\cdot)$. In the following plots, we observe that with each successive iteration, the disturbance supports shrink, the terminal sets enlarge and the trajectory costs decrease.

\begin{itemize}
    \item \textbf{Disturbance Bound:} For Step 1 of Algorithm~\ref{algo:ROLMPC} at iteration $j$, we use the approach in Theorem~\ref{thm:D_cnstrct} of Section~\ref{ssec:err_inv} with system trajectory data $\{\{x^i_t,u^i_t\}_{t\geq 0}\}_{i=0}^{j-1}$ to obtain the outer-approximation of the disturbance support $\hat{\mathcal{D}}^j$ as a product of intervals, $\hat{\mathcal{D}}^j_1\times \hat{\mathcal{D}}^j_2\times\hat{\mathcal{D}}_3^j\subset\mathbb{R}^3$. We show the constructed disturbance support for iterations $j=1,10,15,20, 25$ in Figure~\ref{fig:W}. Notice that the disturbance supports shrink as more trajectory data as collected, as enforced by the SDP \eqref{eq:LMI0}.
\begin{figure}[h]
    \centering
\includegraphics[width=\linewidth]{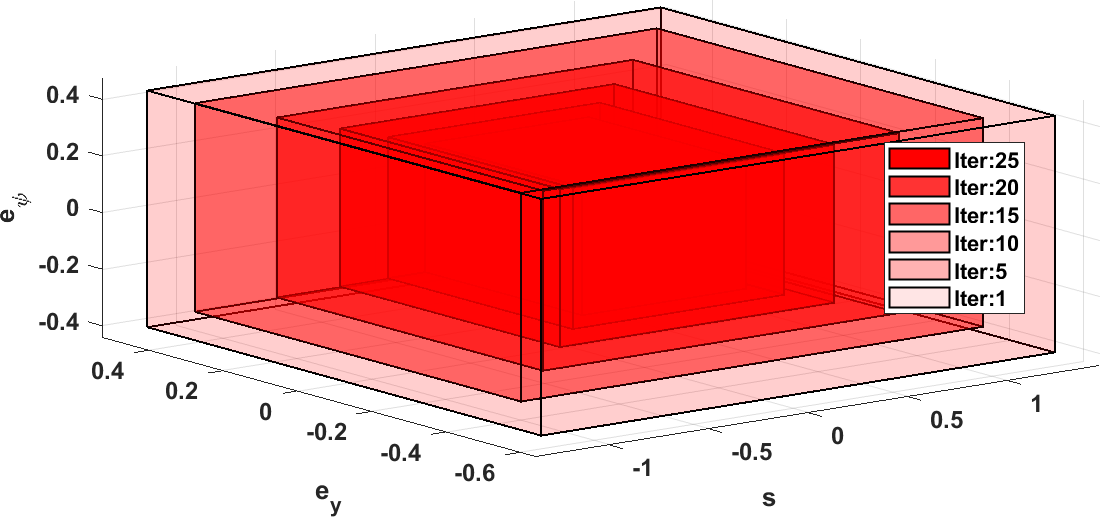}
    \caption{{\small{Disturbance support estimates across iterations. As more data is collected, the support estimates shrink.}}}
    \label{fig:W}
\end{figure}
\item \textbf{Tightened Constraints:} For Step 2 of Algorithm~\ref{algo:ROLMPC} at iteration $j$, the tightened constraints $\bar{\mathcal{X}}^j,\bar{\mathcal{U}}^j$ are obtained by solving NLP \eqref{eq:tight_const_NLP} in Section~\ref{ssec:constraints}. This requires $\bar{\mathcal{X}}^{j-1},\bar{\mathcal{U}}^{j-1}$, the error invariant $\mathcal{E}^j$ and error policy $\kappa(e_t)=Ke_t$ for the error dynamics \eqref{eq:error_dyn} with disturbance support $\hat{\mathcal{D}}^j$ computed in Step 1. The error invariant and error policy are computed as described in Section~\ref{ssec:kb_prob_form}. We show the tightened constraints for iterations $j=1,10,15,20, 25$ in Figure~\ref{fig:XtUt}. Notice that the tightened state and input constraints increase in size with increasing iterations, as guaranteed by Proposition~\ref{prop:tight_cons_NLP}.
\begin{figure}[h]
\begin{subfigure}{.5\textwidth}
  \centering
  \includegraphics[width=.95\linewidth]{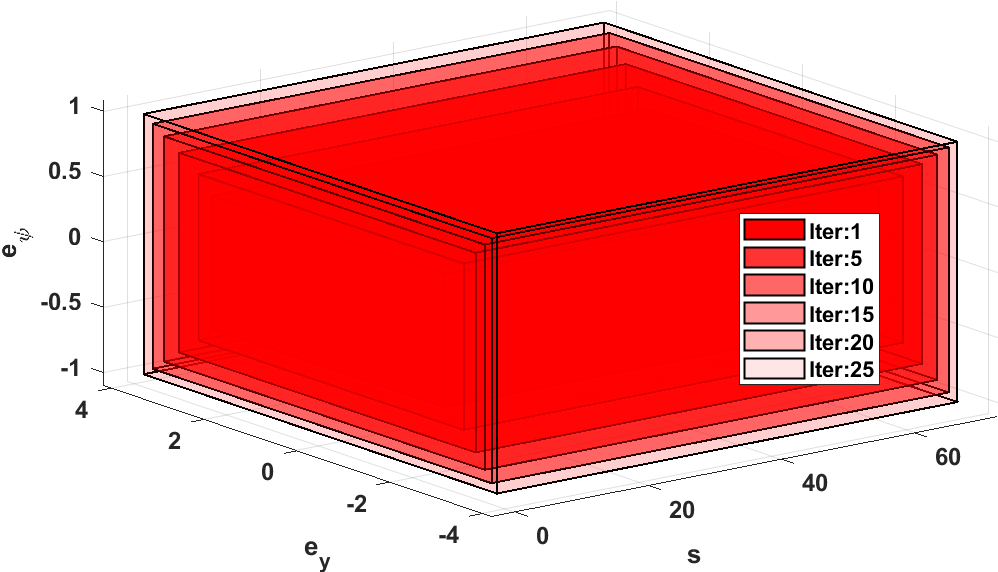}
  \caption{{\small{Tightened state constraints across iterations.}}}
  \label{fig:Xt}
\end{subfigure}%

\begin{subfigure}{.5\textwidth}
  \centering
  \includegraphics[width=.95\linewidth]{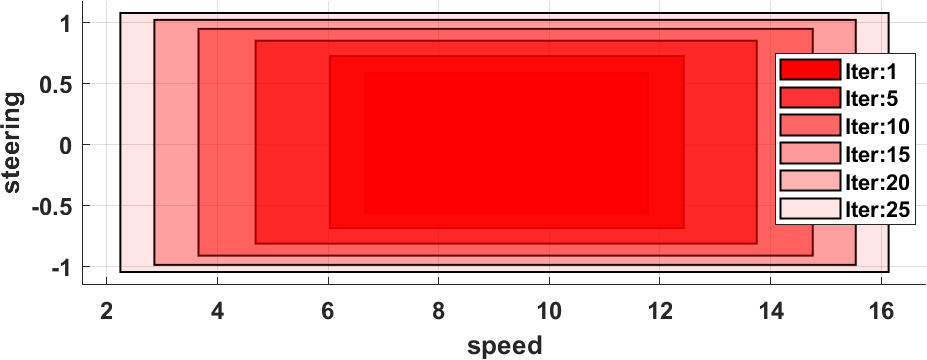}
  \caption{{\small{Tightened input constraints across iterations.}}}
  \label{fig:Ut}
\end{subfigure}
\caption{{\small{The tightened constraints increase in size across iterations, as the model uncertainty is learned.}}}
\label{fig:XtUt}
\end{figure}
\item \textbf{Terminal Set:}  The terminal set at iteration $j$ is constructed from  $\bar{\mathcal{CS}}^j_{\mathbf{y}}$ as  $\bar{\mathcal{X}}^j_N=\{\bar{x}| \exists\bar{\mathbf{y}}\in\bar{\mathcal{CS}}^j_{\mathbf{y}}, \mathcal{F}_x(\bar{\mathbf{y}})=\bar{x}\}$. To visualize this set, 1) we sample points in $\bar{\mathcal{CS}}^j_{\mathbf{y}}$, 2) map them onto the state space via $\mathcal{F}_x(\cdot)$ and 3) use Matlab's \textit{alphaShape} function to fit a surface over the projected points (note that the image of convex set $\bar{\mathcal{CS}}^j_{\mathbf{y}}$ under continuous, surjective function $\mathcal{F}_x(\cdot)$ can be shown to be path-connected \cite[Chapter 9]{munkres2000topology}). The resulting terminal set approximation is shown in Figure~\ref{fig:term} using trajectory data up to iteration $j=25$, in both Frenet and global coordinates. Note that these continuous sets were constructed without any local linear approximations, or a pre-computed reference, and utilise the complete nonlinear dynamics of \eqref{eq:sysdyn} \textit{implicitly} via trajectory data $\{\bar{x}_t, \bar{u}_t\}_{t\geq 0}$ and the map $\mathcal{F}_x(\cdot)$. 

\begin{figure}[!h]
\centering
\begin{subfigure}[b]{0.47\columnwidth}
   \includegraphics[width=\linewidth]{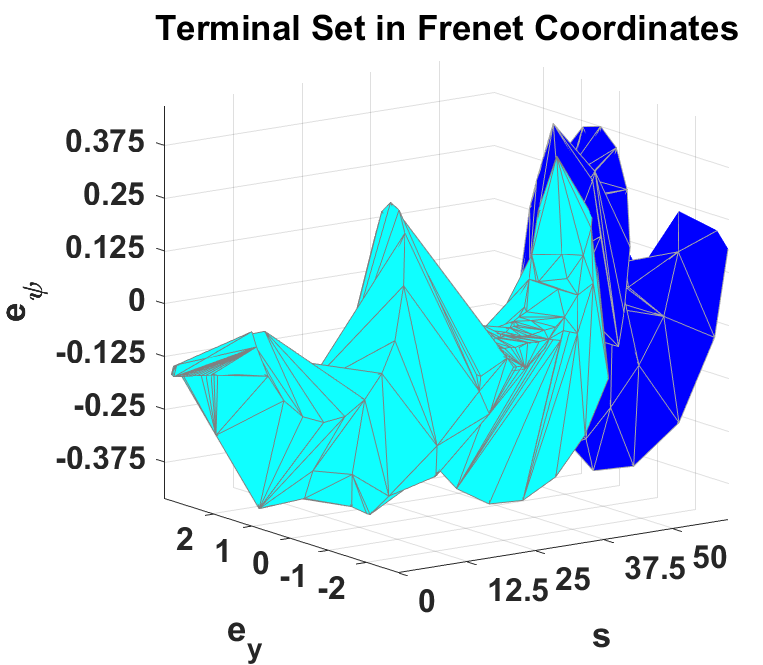}
\end{subfigure}
\begin{subfigure}[b]{0.47\columnwidth}
   \includegraphics[width=\linewidth]{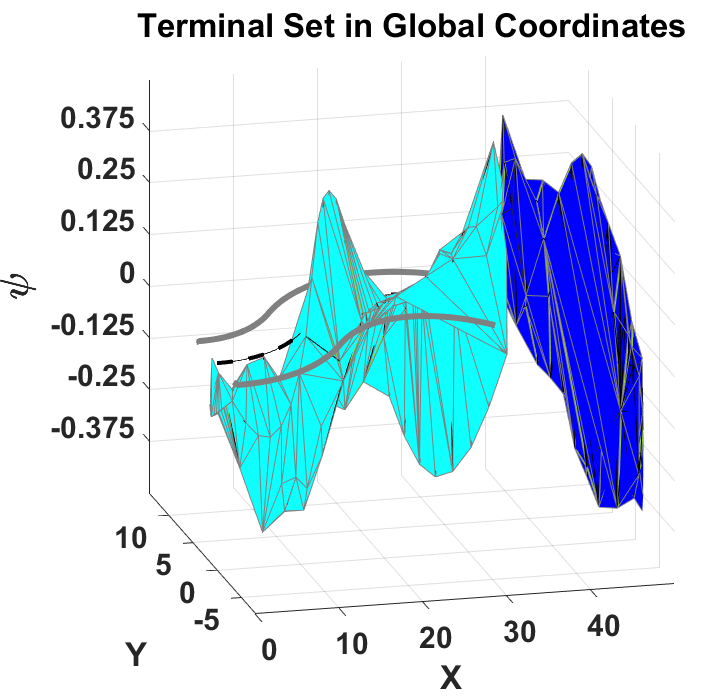}
\end{subfigure}
\caption{{\small{Terminal sets in state space, in Frenet and global coordinates constructed from nominal system trajectory data up to iteration $j=20$. The dark blue regions denote   states in $\mathcal{X}_G$. }}}\label{fig:term}
\end{figure}
\item\textbf{Trajectory Costs:} We plot the trajectory costs of the closed-loop trajectories across all the iterations in Figure~\ref{fig:traj_cost}, and see that the trajectory costs decrease with each iteration, validating the claim of Theorem~\ref{thm:cost_imp}. 
\begin{figure}[!h]
    \centering
    \includegraphics[width=1\linewidth]{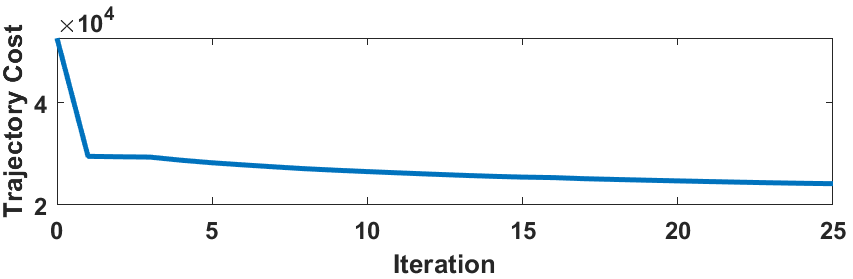}
    \caption{{\small{Closed-loop trajectory costs decrease across iterations.}}}
    \label{fig:traj_cost}
\end{figure}
\end{itemize}

    

    

\subsubsection{Robust Constraint Satisfaction}: The tightened constraints within the Robust MPC formulation ensure that the closed-loop system trajectories satisfy the constraints robustly, despite the uncertainty in the dynamics. 

We plot the closed-loop state and input trajectories of the kinematic bicycle in global coordinates in Figures~\ref{fig:traj_bi}, \ref{fig:speed_bi}, \ref{fig:delta_bi} across the iterations. Iteration 0 corresponds to the first trajectory with which our algorithm was initialized. At iteration 25, we see that the path of the closed-loop trajectory is significantly tighter than that of iteration 0 in Figure~\ref{fig:traj_bi}. From the Figures~\ref{fig:speed_bi},~\ref{fig:nspeed_bi}, \ref{fig:delta_bi}, we see that the actual and nominal speed profiles, and the steering commands are within constraints.  The tightened constraints for the nominal speed \ref{fig:nspeed_bi} are shown for $j=25$. Also notice in Figure~\ref{fig:nspeed_bi} that the trajectory in iteration 25 reaches $\mathcal{X}_G$ the fastest. 
\begin{figure}[!h]
    \centering
    \includegraphics[width=0.9\linewidth]{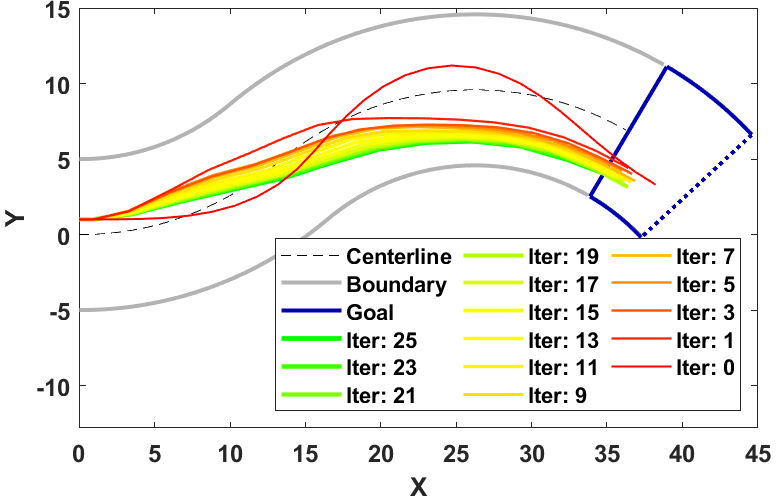}
    \caption{{\small{State trajectories across iterations.}}}
    \label{fig:traj_bi}
\end{figure}
\begin{figure}[!h]
\centering
\begin{subfigure}[b]{0.5\textwidth}
   \includegraphics[width=1\linewidth]{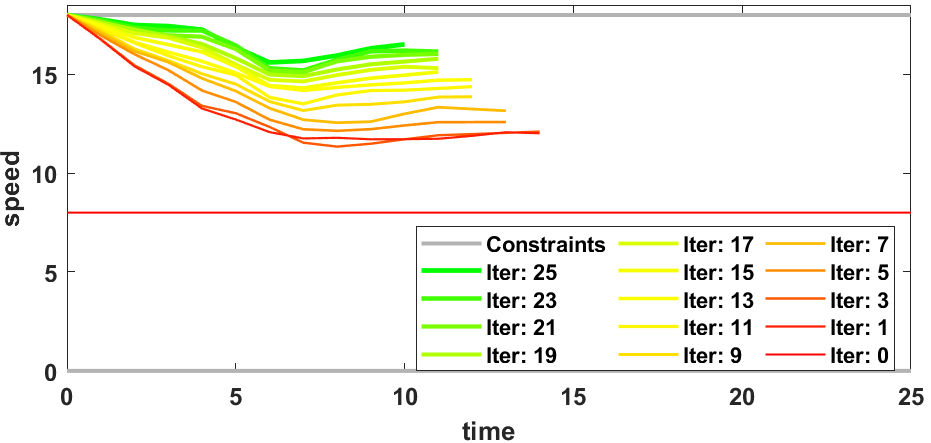}
   \caption{Speed profile}
   \label{fig:speed_bi} 
\end{subfigure}
\begin{subfigure}[b]{0.5\textwidth}
   \includegraphics[width=1\linewidth]{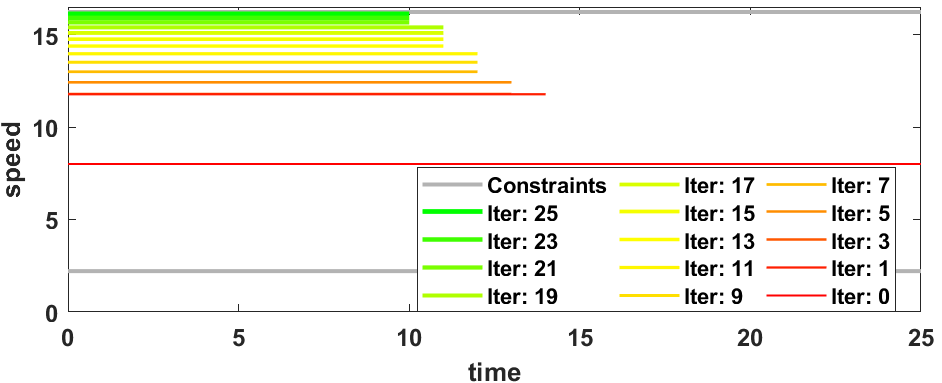}
   \caption{Nominal speed profile}
   \label{fig:nspeed_bi} 
\end{subfigure}
\begin{subfigure}[b]{0.45\textwidth}
   \includegraphics[width=1\linewidth]{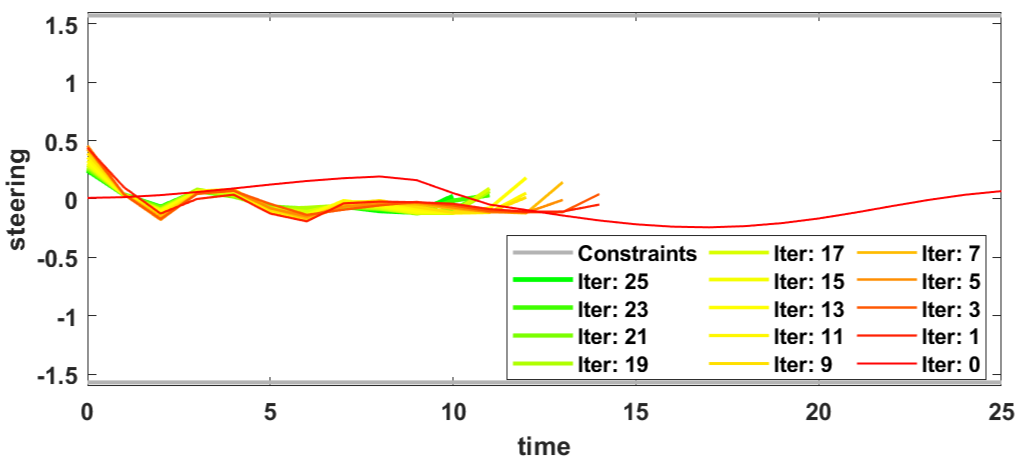}
   \caption{Steering profile}
   \label{fig:delta_bi}
\end{subfigure}
\caption{{\small{Input trajectories across iterations.}}}
\end{figure}

\subsubsection{Computational Tractability}: The proposed Convex Output Safe set $\bar{\mathcal{CS}}^j_{\mathbf{y}}$ is a convex set, as opposed to the discrete Safe set construction $\mathcal{SS}^j$ in \cite{UgoTAC}. This improves computation efficiency for solving the optimization problem \ref{eq:OP_RLMPC} without sacrificing performance guarantees.

We compare the average solve times for our approach and the LMPC from \cite{UgoTAC} to demonstrate the benefit of using the continuous safe set $\bar{\mathcal{CS}}^j_{\mathbf{y}}$ over the discrete safe set $\mathcal{SS}^j$. The former leads to solving a nonlinear program (which is solved using IPOPT) and the latter requires solving a mixed-integer nonlinear program (which is solved using BONMIN). In Figure~\ref{fig:comp_cost}, we see that the solve times increase with iterations because of the growing size of the safe sets, but our approach is markedly more efficient (with solve times $\leq 10^{-0.5} s\approx 0.3 s$).
\begin{figure}[!h]
    \centering
    \includegraphics[width=1.1\linewidth]{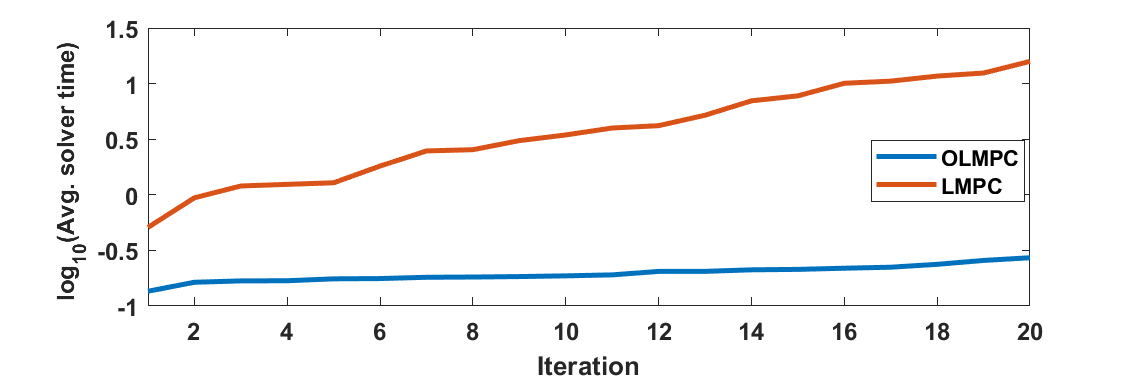}
    \caption{{\small{Avg. solve times for Output-lifted LMPC and LMPC for iterations $j=1$ to $j=20$ in $\log_{10}$ scale.}}}
    \label{fig:comp_cost}
\end{figure}
\section{Conclusion}
We have proposed a formulation of Robust LMPC for systems with lifted outputs performing iterative tasks. We showed that with certain properties of these outputs, we can iteratively construct continuous, control invariant terminal sets and CLF terminal costs for nonlinear system dynamics using historical state-input trajectory data. Furthermore, we show how to use the trajectory data to quantify and iteratively decrease model uncertainty, and construct tightened state and input constraints for the Robust MPC design. The proposed Robust Output-lifted LMPC scheme is recursively feasible, convergent and iteratively improves system performance while guaranteeing robust constraint satisfaction.

\bibliographystyle{ieeetr}
\bibliography{root.bib}
\section{Appendix}
\subsection{Proof of Proposition \ref{prop:qc}}
By Assumption~\ref{ass:dQC}, we know that for any $q,q'\in\mathbb{R}^{n+m}$, $z\in D(q),z'\in D(q')$, we have the incremental inequalities
\begin{align*}
     &\begin{bmatrix}
     1\\q-q'\\z-z'
     \end{bmatrix}^\top Q^{(j)}\begin{bmatrix}
     1\\q-q'\\ z-z'
     \end{bmatrix}\geq 0,~\forall Q^{(j)}\in\boldsymbol{\mathcal{Q}},
 \end{align*}
 which also holds for $(q,z), (q',z')\in G(D)$ by Definition \ref{def:graph}. Now define the set $\mathcal G^i_t$ using the incremental inequalities and $(q^i_t,d^i_t)\in G(D)$ as in \eqref{eq:E_approx_jt}. Observe that for any $(q,z)\in G(D)$, we have $(q,z)\in\mathcal{G}^i_t$. Thus, $G(D)\subseteq\mathcal{G}^i_t$. Since this holds for any iteration $i$ and time $t$ for system \eqref{eq:sysdyn}, we have $G(D)\subseteq\mathcal{G}^i_t,~\forall t\geq 0, \forall i\leq j-1$ which further implies $G(D)\subseteq\cap_{i=0}^{j-1}\cap_{t\geq 0}\mathcal{G}^i_t$
\hfill$\blacksquare$
\subsection{Proof of Proposition \ref{prop:tight_cons_NLP}}
\begin{enumerate}
\item If \eqref{eq:tight_const_NLP} is feasible for $j=1$, we have $\bar{\mathcal{X}}^{0}\subseteq\bar{\mathcal{X}}^{1}, \bar{\mathcal{U}}^{0}\subseteq\bar{\mathcal{U}}^{1}$. For $j>1$, see that since $\mathcal{E}^j\subset\mathcal{E}^{j-1}$, we have $\hat{\mathcal{X}}^{j}\supset \hat{\mathcal{X}}^{j-1}$ and $\hat{\mathcal{U}}^{j}\supset\hat{\mathcal{U}}^{j-1}$. Thus, $\bar{\mathcal{S}}^j_x\supseteq \bar{\mathcal{S}}^{j-1}_x, \bar{\mathcal{S}}^j_u\supseteq \bar{\mathcal{S}}^{j-1}_u$ and since $\bar{\mathcal{X}}^{j-1}\subseteq\bar{\mathcal{S}}^{j-1}_x$, $\bar{\mathcal{U}}^{j-1}\subseteq\bar{\mathcal{S}}^{j-1}_u$, we also have $\bar{\mathcal{X}}^{j-1}\subseteq\bar{\mathcal{S}}^{j}_x$, $\bar{\mathcal{U}}^{j-1}\subseteq\bar{\mathcal{S}}^{j}_u$. This implies that $\alpha^j_x=1, \alpha^j_u=1$, $v^j_x=0, v^j_u=0$ is feasible for \eqref{eq:tight_const_NLP}.
\item Since $\bar{\mathcal{S}}_x^j\subset \hat{\mathcal{X}}^j=\mathcal{X}\ominus \mathcal{E}^j, \bar{\mathcal{S}}_u^j\subset \hat{\mathcal{U}}^j=\mathcal{U}\ominus K\mathcal{E}^j $ by definition, and $\bar{\mathcal{X}}^j\subseteq\bar{\mathcal{S}}_x^j, \bar{\mathcal{U}}^j\subseteq\bar{\mathcal{S}}_u^j$ by feasibility of \eqref{eq:tight_const_NLP}, we have $\bar{x}_t\in\bar{\mathcal{X}}^{j}, \bar{u}_t\in\bar{\mathcal{U}}^{j}\Rightarrow x_t\in\mathcal{X}, u_t\in\mathcal{U}$.
\item Before proving the statement, we first prove the following auxiliary property that is granted by Assumption \ref{ass:flat}\eqref{ass:flatmap_convex}:
\small
\begin{align*}
  [\mathcal{F}^{\cap}(\bar{\mathbf{Y}}), \mathcal{F}^{\cup}(\bar{\mathbf{Y}})]\subseteq[\min_{\substack{k=1,.., p}}\mathcal{F}^{\cap}(\bar{\mathbf{Y}}^k),&\max_{\substack{k=1,.., p}}\mathcal{F}^{\cup}(\bar{\mathbf{Y}}^k)]
\end{align*}\normalsize
for $\bar{\mathbf{Y}}\in\textrm{conv}(\{\bar{\mathbf{Y}}^1,\dots,\bar{\mathbf{Y}}^p\})$, where the intervals and $\min, \max$ are defined elment-wise. 
We proceed using induction on $p$, the number of points in the set. For $p=2$, the property follows trivially by Assumption \ref{ass:flat}\eqref{ass:flatmap_convex}.
Suppose the property is true for $p-1$, i.e.,\small
\begin{align*}
&[\mathcal{F}^{\cap}(\bar{\mathbf{Y}}'), \mathcal{F}^{\cup}(\bar{\mathbf{Y}}')]\subseteq\\
&[\min_{\substack{k=1,.., p-1}}\mathcal{F}^{\cap}(\bar{\mathbf{Y}}^k),\max_{\substack{k=1,.., p-1}}\mathcal{F}^{\cup}(\bar{\mathbf{Y}}^k)]
\end{align*}\normalsize
for any $\bar{\mathbf{Y}}' \in\textrm{conv}(\{\bar{\mathbf{Y}}^1,\dots,\bar{\mathbf{Y}}^{p-1}\})$. Adding an additional point in the set, let $\textrm{conv}(\{\bar{\mathbf{Y}}^1,\dots,\bar{\mathbf{Y}}^p\})\ni\bar{\mathbf{Y}}=\lambda\bar{\mathbf{Y}}^p+(1-\lambda)\bar{\mathbf{Y}}'$ for some $\lambda\in[0,1]$. Using the property for $p=2$, we have \small
\begin{align*}
&[\mathcal{F}^{\cap}(\bar{\mathbf{Y}}), \mathcal{F}^{\cup}(\bar{\mathbf{Y}})]\subseteq\\&[\min(\mathcal{F}^{\cap}(\bar{\mathbf{Y}}^p),\mathcal{F}^{\cap}(\bar{\mathbf{Y}}')),\max(\mathcal{F}^{\cup}(\bar{\mathbf{Y}}^p),\mathcal{F}^{\cup}(\bar{\mathbf{Y}}'))]
\end{align*}\normalsize
Using the truth of property for $p-1$, we therefore write
\small
\begin{align*}
&\min_{\substack{k=1,.., p}}\mathcal{F}^{\cap}(\bar{\mathbf{Y}}^k)\leq\min(\mathcal{F}^{\cap}(\bar{\mathbf{Y}}^p),\mathcal{F}^{\cap}(\bar{\mathbf{Y}}')),\\ 
&\max(\mathcal{F}^{\cup}(\bar{\mathbf{Y}}^p),\mathcal{F}^{\cup}(\bar{\mathbf{Y}}')) \leq \max_{\substack{k=1,.., p}}\mathcal{F}^{\cup}(\bar{\mathbf{Y}}^k)\\
\Rightarrow &[\mathcal{F}^{\cap}(\bar{\mathbf{Y}}), \mathcal{F}^{\cup}(\bar{\mathbf{Y}})]\subseteq[\min_{\substack{k=1,\dots, p}}\mathcal{F}^{\cap}(\bar{\mathbf{Y}}^k),\max_{\substack{k=1,.., p}}\mathcal{F}^{\cup}(\bar{\mathbf{Y}}^k)]\ (\star)
\end{align*}\normalsize
where the $\min, \max$ for vectors are computed element-wise. The property thus holds true for $p$ as well and induction helps us conclude that this holds for any $p\geq 1$. 

Now we prove statement 3) of the proposition. We have $\forall k=1,..,p$, $\mathcal{F}(\bar{\mathbf{Y}}^k)=(\bar{x}^k,\bar{u}^k)\in \bar{\mathcal{X}}^j \times \bar{\mathcal{U}}^j$. Since $\bar{\mathcal{X}}^j \times \bar{\mathcal{U}}^j\subset \bar{\mathcal{S}}_x^j\times\bar{\mathcal{S}}_u^j$ by construction, this implies that $\mathcal{F}^{\cup}(\bar{\mathbf{Y}}^k), \mathcal{F}^{\cap}(\bar{\mathbf{Y}}^k)\in\hat{\mathcal{X}}^j\times\hat{\mathcal{U}}^j~\forall k=1,\dots, p$\footnote{Technically, $\exists \bar{\mathbf{Y}}^{'k}: \mathcal{F}(\bar{\mathbf{Y}}^{'k})=(\bar{x}^k,\bar{u}^k), \mathcal{F}^{\cup}(\bar{\mathbf{Y}}^{'k}), \mathcal{F}^{\cap}(\bar{\mathbf{Y}}^{'k})\in\hat{\mathcal{X}}^j\times\hat{\mathcal{U}}^j$ but due to the uniqueness of $(\bar{x}^k,\bar{u}^k)$ granted by Definition \ref{def:diffFlat}, we can use $\bar{\mathbf{Y}}^k$ instead of $\bar{\mathbf{Y}}^{'k}$ w.l.o.g.} . Additionally $\hat{\mathcal{X}}^j\times\hat{\mathcal{U}}^j$ is a box constraint, so we have $[\min\limits_{k=1,\dots, p}\mathcal{F}^{\cap}(\bar{\mathbf{Y}}^k),\max\limits_{k=1,\dots, p}\mathcal{F}^{\cup}(\bar{\mathbf{Y}}^k)]\subseteq\hat{\mathcal{X}}^j\times\hat{\mathcal{U}}^j$. Finally by using result $(\star)$ and $\mathcal{F}^\cap(\bar{\mathbf{Y}})\leq \mathcal{F}(\bar{\mathbf{Y}})\leq\mathcal{F}^\cup(\bar{\mathbf{Y}})$, we have $\mathcal{F}(\bar{\mathbf{Y}})\in \bar{\mathcal{X}}^j \times \bar{\mathcal{U}}^j$ for any $\bar{\mathbf{Y}}\in\textrm{conv}(\{\bar{\mathbf{Y}}^1,\dots,\bar{\mathbf{Y}}^p\})$.
\end{enumerate}
\hfill$\blacksquare$
 
\subsection{Proof of Proposition \ref{prop:CS_CI}}
By definition of $\bar{\mathcal{CS}}_{\mathbf{y}}^{j}$ we have for $\bar{\mathbf{y}}_t\in\bar{\mathcal{CS}}_{\mathbf{y}}^{j}$,\small
\begin{align}\label{eq:y_convCS}
&\bar{\mathbf{y}}_t=\sum\limits_{i=0}^{j-1}\sum\limits_{k\geq 0}\lambda^i_k\bar{\mathbf{y}}^i_k,\ \bar{\mathbf{y}}^i_k\in\bar{\mathcal{CS}}_{\mathbf{y}}^{j}\subset\mathbb{R}^{m\times R}\\
&\sum_{i=0}^{j-1}\sum_{k\geq 0}\lambda^i_k=1,\ \lambda^i_k\geq 0\nonumber
\end{align}\normalsize
By the definition of $\bar{\mathcal{CS}}_{\mathbf{y}}^{j}$ in \eqref{eq:CSy_def}, each $\bar{\mathbf{y}}^i_k$ maps to a feasible state, i.e., $\mathcal{F}_x(\bar{\mathbf{y}}^i_k)=\bar{x}^i_k\in\bar{\mathcal{X}}^{i+1}\subseteq\bar{\mathcal{X}}^{j}$. Invoking Proposition~\ref{prop:tight_cons_NLP}(3) gives us, \small
\begin{equation}\label{eq:imp_1_CI}
    \mathcal{F}_x(\bar{\mathbf{y}}_t)=\bar{x}_t\in\bar{\mathcal{X}}^j
\end{equation}\normalsize
See that $\bar{\mathbf{y}}^i_k\in\bar{\mathcal{CS}}_{\mathbf{y}}^{j}\Rightarrow\bar{\mathbf{y}}^i_{k+1}\in\mathcal{CS}_{\mathbf{y}}^{j}$. We use the lifted-output $\bar{\mathbf{Y}}^i_k$ and map $\mathcal{F}_u(\cdot)$ to reconstruct the nominal input as $\bar{u}^i_k=u^i_k-K(x^i_k-\bar{x}^i_k)=\mathcal{F}_u([\bar{y}^i_{k},\bar{y}^i_{k+1},\dots,\bar{y}^i_{k+R}])=\mathcal{F}_u([\bar{\mathbf{y}}^i_{k},\bar{y}^i_{k+R}])$ and note that $\bar{u}^i_k\in\bar{\mathcal{U}}^{i+1}\subset\bar{\mathcal{U}}^j$ from \eqref{eq:nom_b}. Consider the following control input
\small
\begin{align}\label{eq:safe_u}
    \bar{u}_t&=\mathcal{F}_u(\sum\limits_{i=0}^{j-1}\sum\limits_{k\geq 0} \lambda^i_k[\bar{\mathbf{y}}^i_{k},\bar{y}^i_{k+R}])\nonumber\\
    &=\mathcal{F}_u([\bar{\mathbf{y}}_t, \bar{y}_{t+R}])
\end{align}\normalsize
where $\bar{y}_{t+R}=\sum_{i=0}^{j-1}\sum_{k\geq 0} \lambda^i_k \bar{y}^i_{k+R}$. Invoking Proposition~\ref{prop:tight_cons_NLP}(3) again proves $\bar{u}_t\in\bar{\mathcal{U}}^j$. Also see that 
\begin{align}\label{eq:time_shift_step}\small
    \bar{\mathbf{y}}_{t+1}&=\delta([\bar{\mathbf{y}}_t,\bar{y}_{t+R}])\nonumber\\
    &=\delta(\sum\limits_{i=0}^{j-1}\sum\limits_{k\geq 0} \lambda^i_k[\bar{\mathbf{y}}^i_k,\bar{y}^i_{k+R}])\nonumber\\
    &=\sum\limits_{i=0}^{j-1}\sum\limits_{k\geq 0} \lambda^i_k \bar{\mathbf{y}}^i_{k+1}\Rightarrow \bar{\mathbf{y}}_{t+1}\in\bar{\mathcal{CS}}_{\mathbf{y}}^{j}.
\end{align}\normalsize 
Let $\bar{u}_2,\dots,\bar{u}_{R-1}\in\mathbb{R}^m$ be the remaining inputs that generate $[\bar{y}_t,\bar{\mathbf{y}}_{t+1}]\in\mathbb{R}^{m\times R+1}$, i.e.,\small
\begin{align}\label{eq:flat_seq_gen}
    [\bar{y}_t,\bar{\mathbf{y}}_{t+1}]=&[h(\bar{x}_t),h(f(\bar{x}_t,\bar{u}_t)),h(f^{(2)}(\bar{x}_t,\bar{u}_t,\bar{u}_2)),\dots,\nonumber\\
    &h(f^{(R-1)}(\bar{x}_t,\bar{u}_t,\dots,\bar{u}_{R-1}))]\in\mathbb{R}^{m\times R+1}
\end{align}\normalsize
where \small$f^{(k)}(\bar{x}_t,\bar{u}_t,\dots,\bar{u}_k)=\underbrace{f(\dots(f}_{k\textrm{ times}}(\bar{x}_t,\bar{u}_t),\dots \bar{u}_k).$\normalsize
 Using the map \eqref{eq:flatclass_x} to construct the nominal state, we can write\small
\begin{align*}
    \mathcal{F}_x(\bar{\mathbf{y}}_{t+1})=&\mathcal{F}_x([h(f(\bar{x}_t,\bar{u}_t)),h(f^{(2)}(\bar{x}_t,\bar{u}_t,\bar{u}_2)),\dots,\\&
    h(f^{(R-1)}(\bar{x}_t,\bar{u}_t,\dots,\bar{u}_{R-1}))])\\
    =&f(\bar{x}_t,\bar{u}_t)
\end{align*}\normalsize
where the last equality is true because of the unique correspondence from \small$[\bar{y}_t,\dots, \bar{y}_{t+R-1}]=[h(\bar{x}_t),\dots, h(f^{(R-1)}(\bar{x}_t,\bar{u}_t,\dots,\bar{u}_{t+R-1}))]$\normalsize  to $\bar{x}_t$ (Definition~\ref{def:diffFlat}). Finally, invoking Proposition~\ref{prop:tight_cons_NLP}(3) again using sequences $\bar{\mathbf{y}}^i_{k+1}, \forall i=0,\dots,j-1$ gives us\small
\begin{equation}\label{eq:imp_3_CI}
f(\bar{x}_t,\bar{u}_t)\in\bar{\mathcal{X}}^j.
\end{equation}\normalsize
\hfill$\blacksquare$
\subsection{Proof of Proposition~\ref{prop:Q_CLF}}
1) First note that $\bar{\mathbf{y}}\in\bar{\mathcal{CS}}_{\mathbf{y}}^{j}$ implies that the optimization problem implicit in the definition \eqref{eq:conv_cf} of $\bar{Q}^{j}(\cdot)$ is feasible. Also see that since the feasible set is compact (countable product of compact sets is compact by Tychonoff's theorem)  and the objective is continuous (linear, in fact, and bounded because of Theorem ), a minimizer exists by Weierstrass' theorem for every $\bar{\mathbf{y}}\in\bar{\mathcal{CS}}_{\mathbf{y}}^{j}$. Thus for any $\bar{\mathbf{y}}\in\bar{\mathcal{CS}}_{\mathbf{y}}^{j}$, we can write 
$\bar{Q}^{j}(\bar{\mathbf{y}})= \sum\limits_{i=0}^{j-1}\sum\limits_{k\geq 0}\lambda^{\star i}_k\mathcal{C}^i_k$ where the $\lambda^{\star i}_k$s satisfy the constraints in \eqref{eq:conv_cf}. The definition of $\mathcal{C}^i_k$ in \eqref{eq:ctg} and positive definiteness of $c(\cdot)$ by \eqref{stage_cost} imply that $\bar{Q}^{j}(\mathbf{y}) \succ 0 \ \forall \mathbf{y}\in\mathcal{CS}_{\mathbf{y}}^{j}\backslash\mathcal{Y}_G$. For any $\bar{\mathbf{y}}\in\mathcal{Y}^j_G$, we have $\bar{\mathbf{y}}=\sum_{i=0}^{j-1}\sum_{k\geq 0}\lambda_k^i\bar{\mathbf{y}}_k^i$ with $\lambda^i_k>0$ only for $\bar{\mathbf{y}}^i_k\in\mathcal{Y}_G$, implying $\bar{Q}^j(\bar{\mathbf{y}})\leq 0$. Thus, $\bar{Q}^{j}(\mathbf{y})=0~\forall \bar{\mathbf{y}}\in\mathcal{Y}^j_G$.   We finish the proof for the first part by observing that for $\bar{\mathbf{y}}\in\mathcal{Y}_G\backslash\mathcal{Y}^j_G$, there exists no combination of multipliers such that $\lambda_k^i>0$ only for $\bar{\mathbf{y}}^i_k\in\mathcal{Y}_G$, and since $\mathcal{C}^i_k>0$ for $\bar{\mathbf{y}}^i_k\not\in\mathcal{Y}_G$, we must have $\bar{Q}^j(\bar{\mathbf{y}})>0$.\\
2) For any $\bar{\mathbf{y}}_t\in\bar{\mathcal{CS}}_{\mathbf{y}}^{j}$, let $\bar{Q}^{j}(\bar{\mathbf{y}}_t)= \sum\limits_{i=0}^{j-1}\sum\limits_{k\geq0}\lambda^{\star i}_k\mathcal{C}^i_k$ with $\lambda^{\star i}_k$ satisfying the constraints in \eqref{eq:conv_cf}. Observing the linearity of the forward-time shift operator $\delta(\cdot,\cdot)$, we have
\small
\begin{align*}
    \bar{\mathbf{y}}_{t+1}&=\delta(\bar{\mathbf{y}}_t,\bar{y}_{t+R})\\
    &=\delta(\sum\limits_{i=0}^{j-1}\sum\limits_{k\geq0}\lambda^{\star i}_k\bar{\mathbf{y}}^i_k, \sum\limits_{i=0}^{j-1}\sum\limits_{k\geq0}\lambda^{\star i}_k \bar{y}^i_{k+R})\\
    &=\sum\limits_{i=0}^{j-1}\sum\limits_{k\geq0}\lambda^{\star i}_k\delta(\bar{\mathbf{y}}^i_k,\bar{y}^i_{k+R})\\
    &=\sum\limits_{i=0}^{j-1}\sum\limits_{k\geq0}\lambda^{\star i}_k\bar{\mathbf{y}}^i_{k+1}.
a\end{align*}\normalsize
Thus the same $\lambda^{\star i}_k$s are also feasible for \eqref{eq:conv_cf} at $\mathbf{y}_{t+1}$ and we have
\small
\begin{align*}
    \bar{Q}^{j}(\bar{\mathbf{y}}_{t+1})-\bar{Q}^{j}(\bar{\mathbf{y}}_t)&\leq \sum\limits_{i=0}^{j-1}\sum\limits_{k\geq0}\lambda^{\star i}_k(\mathcal{C}^i_{k+1}-\mathcal{C}^i_{k})\\
    &= \sum\limits_{i=0}^{j-1}\sum\limits_{k\geq0}\lambda^{\star i}_k(-c(\bar{\mathbf{Y}}^i_k))\\
    &\leq -c(\sum\limits_{i=0}^{j-1}\sum\limits_{k\geq0}\lambda^{\star i}_k \bar{\mathbf{Y}}^i_k)\\
    &=-c(\bar{\mathbf{Y}}_t)
\end{align*}\normalsize
The second to last inequality comes from the convexity of $c(\cdot)$.
This completes the proof of the second part of the proposition.
\hfill$\blacksquare$



\end{document}